\documentclass{article}
\usepackage{ijcai22}

\usepackage{times}
\usepackage{soul}
\usepackage{url}
\usepackage[hidelinks]{hyperref}
\usepackage[utf8]{inputenc}
\usepackage[small]{caption}
\usepackage{graphicx}
\usepackage{amsmath}
\usepackage{amsthm}
\usepackage{booktabs}
\usepackage{algorithm}
\usepackage{algorithmic}
\usepackage{amsthm}
\usepackage{enumerate}
\usepackage{amsmath}
\usepackage{amssymb}
\usepackage{graphicx}
\usepackage{comment}
\usepackage{xcolor}
\usepackage{algorithm}
\usepackage{algorithmic}
\usepackage{newfloat}
\usepackage{listings}
\usepackage{paralist}

\def\ve#1{\mathchoice{\mbox{\boldmath$\displaystyle\bf#1$}}
	{\mbox{\boldmath$\textstyle\bf#1$}}
	{\mbox{\boldmath$\scriptstyle\bf#1$}}
	{\mbox{\boldmath$\scriptscriptstyle\bf#1$}}}

\newcommand\vep{{\ve p}}

\newcommand\vex{{\ve x}}
\newcommand\vey{{\ve y}}
\newcommand\vez{{\ve z}}

\newtheorem{lemma}{Lemma}
\newtheorem{theorem}{Theorem}

\newtheorem{definition}{Definition}
\newtheorem{definition*}{Definition}
\newtheorem{claim}{Claim}
\newtheorem{observation}{Observation}

\title{Local Differential Privacy Meets Computational Social Choice - \\Resilience under Voter Deletion}

\author{
Liangde Tao$^1$
\and
Lin Chen$^2$\footnote{Corresponding Author.}\and
Lei Xu$^3$\And
Weidong Shi$^4$
\affiliations
$^1$Zhejiang University, China\\
$^2$Texas Tech University, US\\
$^3$Kent State University, US\\
$^4$University of Houston, US
\emails
\{vast.tld,chenlin198662,xuleimath\}@gmail.com,
larryshi@ymail.com
}

\begin{document}

\maketitle

\begin{abstract}
The resilience of a voting system has been a central topic in computational social choice. Many voting rules, like plurality, are shown to be vulnerable as the attacker can target specific voters to manipulate the result. What if a local differential privacy (LDP) mechanism is adopted such that the true preference of a voter is never revealed in pre-election polls? In this case, the attacker can only infer stochastic information about a voter's true preference, and this may cause the manipulation of the electoral result significantly harder. The goal of this paper is to provide a quantitative study on the effect of adopting LDP mechanisms on a voting system. We introduce the metric PoLDP (power of LDP) that quantitatively measures the difference between the attacker's manipulation cost under LDP mechanisms and that without LDP mechanisms. The larger PoLDP is, the more robustness LDP mechanisms can add to a voting system. We give a full characterization of PoLDP for the voting system with plurality rule and provide general guidance towards the application of LDP mechanisms.
\end{abstract}

\section{Introduction}
In this paper, we consider the attack on an abstract voting system. We compare the manipulation cost under two scenarios of the voting system: the classical voting system, and the voting system where a local differential privacy (LDP) mechanism is adopted. Our goal is to quantify the difference brought by LDP schemes to a voting system under electoral manipulation.

Extensive research has been conducted towards characterizing the resilience of various voting rules (see, e.g., Brandt et al.~\shortcite{brandt2016handbook} for a survey). In many of these research works, the attacker is assumed to know the true preference of voters through pre-election polls. Many voting rules, especially fundamental rules like {plurality}, are found to be vulnerable~\cite{faliszewski2009hard}. Can we enhance the resilience of a voting system by preventing the attacker from learning the preference of voters? Towards this, we observe that LDP mechanisms are a state-of-the-art approach to mitigating the leakage of sensitive information.

LDP mechanism serves as an important method in social science, particularly in voter turnout reports~\cite{rosenfeld2016empirical}. However, most existing research works on LDP focus on the analysis of each user’s privacy, not much is known from the aspect of the system's security (see, e.g., Bebensee~\shortcite{bebensee2019local} as a survey).  Does the introduction of LDP mechanisms make the system more resilient? The most relevant research in this direction is a very recent paper 
\cite{cheu2021}, where they showed that an attacker may manipulate a small percentage of users in an LDP mechanism to mislead the estimation of distribution parameters. However, their attack model is completely different from what we study in this paper.



 We believe characterization of the resilience of a voting system under LDP mechanisms is a natural question that is worth investigating. The high-level description of the specific model we study in this paper is given as follows:

\begin{itemize}
    \item {\it Election setting:} We consider the classical election where there are a set of candidates and a set of independent voters. We focus on the plurality rule (e.g., each voter votes for exactly one candidate). For each voter, we define the type\footnote{For each voter, the type of his/her preference is private information, and will be referred to as the true type to distinguish from the reported type which is generated by the LDP mechanism.} of his/her preference as the candidate he/she votes for. Candidate(s) receiving the highest score will be the co-winner(s). 
    
    \item {\it LDP mechanism:} Every voter will locally and independently run a given LDP mechanism once to generate his/her reported preference which will be used for all pre-election polls. This is a common approach, see e.g., Google's RAPPOR~\cite{erlingsson2014rappor}. Otherwise, the attacker may be able to figure out the voter's true preference via his/her multiple reported preferences~\cite{bebensee2019local}. In this paper, we focus on two fundamental LDP mechanisms: {\it randomized response} and {\it Laplace}.
    \item {\it Attack model:} We consider electoral control by deleting voters. The {\it manipulation cost} is the minimal number of voters the attacker should delete to make the designated candidate win.  
\end{itemize}

We emphasize that the adoption of LDP mechanisms will {\em not} change the electoral result since the mechanisms are only used in pre-election polls. Voters will still vote according to their true preference in the election, but will respond using the reported preference in any pre-election polls. Therefore when the attacker launches the voter deletion attack prior to the election, he/she only knows the reported preference of voters instead of their true preference.

\smallskip
\noindent\textbf{Measuring the impact of LDP mechanisms.} In the classic election setting, the attacker knows the true preference of all voters. Whereas, after adopting the LDP mechanism, the attacker only knows the reported preference. 

Inspired by the concept of PoA (price of anarchy) in algorithmic game theory~\cite{koutsoupias1999worst}, we introduce the metric PoLDP (power of LDP) to quantitatively characterize the resilience brought by the LDP mechanism. Roughly, we can view PoLDP as follows: $$\textrm{PoLDP}=\frac{\textrm{Minimal expected manipulation cost with LDP}}{\textrm{Minimal manipulation cost without LDP}}.$$

If $\textrm{PoLDP}> 1$, then it means the introduction of LDP mechanisms indeed increases the manipulation cost of the attacker, and thus enhances the resilience of the system; if $\textrm{PoLDP}< 1$, then it means the introduction of LDP mechanisms reduces the manipulation cost, and thus diminishes the resilience of the system. It should be noted that the value of $\textrm{PoLDP}$ depends on the specific LDP mechanism as well as the value of the privacy parameter\footnote{See Definition~\ref{def:LDP} for the meaning of the privacy parameter $\epsilon$.}. 

As randomness is involved in an LDP mechanism, the manipulation cost under the LDP mechanism is a random variable, and our definition uses its expected value. One may ask what if this random variable takes a value significantly different from the expectation? Indeed, we will show that since the manipulation cost can be expressed as a summation of independent binary variables, probabilistic analysis guarantees that the manipulation cost as a random variable lies around its expectation with an extremely high probability, 
and consequently our definition of $\textrm{PoLDP}$ by using the expected value is without loss of generality.

\smallskip
\noindent\textbf{Our contributions.}
The main contribution of this paper is to give the first quantitative analysis of the effect brought by local differential privacy in a voting system under electoral control of voter deletion. We study two major LDP mechanisms that are widely adopted, randomized response and Laplace, and show that they can generally enhance the resilience of a voting system. We quantify such an effect through a measure called PoLDP. The larger PoLDP is, the more resilience LDP adds to the system.

Since LDP mechanisms introduce uncertainty, the attacker may need to pay a different manipulation cost to make sure the designated candidate wins with a different probability. That is, the manipulation cost, and consequently PoLDP, is a function of the winning probability.

For specific winning probability, we establish two integer linear programs to give a general upper and lower bound on the manipulation cost. Under some mild assumptions, we show that the upper and lower bound approach to the same value for a big range of the winning probability. That is the manipulation cost for a winning probability of $99.9\%$ is almost the same as that for a winning probability of $0.1\%$. 

Using probabilistic analysis, we give an efficient method to calculate the APoLDP (Asymptotic PoLDP) for voting systems satisfying certain conditions. Furthermore, we study the relationship between the value of APoLDP and the parameter $\epsilon$ of LDP mechanisms. For randomized response and Laplace mechanism, we give the closed-form expression of APoLDP for voting systems with only two candidates. For voting systems with multiple candidates, it becomes too complicated to obtain a simple mathematical formula. Instead, we analyze the maximal value of APoLDP, and the range of $\epsilon$ where APoLDP achieves the maximal value.

Interestingly, we observe that when the parameter $\epsilon$ of LDP mechanisms is below a certain threshold (we call it security threshold), APoLDP stays at its maximal value. Generally, LDP mechanisms with smaller $\epsilon$ can guarantee better privacy. But for manipulation via voter deletion, such kind of extra privacy provided by LDP mechanisms is redundant and can not add more resistance. The security threshold provides general guidance toward the application of LDP mechanisms.

\smallskip
\noindent\textbf{Related works.}
Local differential privacy has been increasingly accepted as a state-of-the-art approach for statistical computations while protecting the privacy of each participant. It was first formalized in~\cite{kasiviswanathan2011can}. Several important LDP mechanisms are studied and compared in the literature~\cite{wang2017locally,bun2019heavy,qin2016heavy}, including the two major LDP mechanisms, randomized response and Laplace considered in this paper. LDP mechanisms have also been adopted by IT companies, e.g., RAPPOR by Google Chrome~\cite{erlingsson2014rappor}. We refer the reader to a comprehensive survey~\cite{bebensee2019local}.

The study of the computational complexity of electoral control was initiated by Bartholdi III et al.~\shortcite{bartholdi1992hard}, who mainly analyzed the voting rule of {plurality}
and {condorcet}, where the attacker's goal is to make a designated candidate win. Later, Hemaspaandra et al.~\shortcite{hemaspaandra2007anyone} studied a closely related model where the attacker's goal is to make the original winner lose. Following their research work, extensive research has been conducted towards understanding the resilience of voting systems under different electoral control methods and voting rules~\cite{hemaspaandra2017complexity,magiera2017hard,maushagen2016complexity,rey2016structural}. While the research has characterized different voting rules as resilient or vulnerable, not much is known about protecting a vulnerable voting rule like {plurality}. {Electoral control under partial information is also considered in the literature~\cite{conitzer2011dominating,dey2018complexity}.} Very recently, Chen et al.~\shortcite{chen2018protecting} and Yin et al.~\shortcite{yin2018optimal} studied the protection of election through deploying defending resources.  

\section{Preliminary}
In this section, we briefly introduce the concept of local differential privacy mechanism. \begin{definition}{$\epsilon$-Local Differential Privacy:} \label{def:LDP}
We say that an mechanism $\mathcal{R}$ satisfies $\epsilon$-local differential privacy where $\epsilon>0$ if and only if for any input $v, v'$ and any $y\in \text{range}(\mathcal{R})$ it holds that $$\frac{\Pr[\mathcal{R}(v)=y]}{\Pr[\mathcal{R}(v')=y]}\leq e^{\epsilon}.$$
\end{definition}

\vspace{1mm}
{We call $\mathcal{R}$ an $\epsilon$-LDP mechanism if it satisfies the $\epsilon$-local differential privacy.} The privacy parameter $\epsilon$ characterizes the level of privacy provided by the LDP mechanism $\mathcal{R}$. For example, perfect privacy is ensured when $\epsilon=0$. And, no privacy is guaranteed when $\epsilon=+\infty$. 

As mentioned before, the type of true preference of each voter (we refer to it as true type) is private information. Voters will vote according to their true types in the election. In the meantime, we let every voter $V_i$ independently and locally run an $\epsilon$-LDP mechanism (which is the randomized response or Laplace in this paper) to generate a type of his/her reported preference (we refer to it as reported type) and use the reported type in all pre-election polls. Consequently, the reported type of each voter is public information.

\smallskip
\noindent\textbf{Design Matrix:}
We now introduce the design matrix of an LDP mechanism. In our problem, the design matrix maps the true type to the reported type for each voter. Consider an arbitrary voter $V_i$. Let $t_i$ be the true type, and $r_i$ be the reported type generated by $\epsilon$-LDP mechanisms. There are $m$ candidates in the voting system. Consequently, there are $m$ different types under {plurality}. Hence, we know that $t_i, r_i\in[1,m]$.

Consider the probability that $V_i$ has a true type $v$ but reports $u$, i.e., the event that its reported type is $u$ conditioned on the event that its true type is $v$, or $\Pr[r_i=u| t_i=v]$. It is clear that since each voter independently and locally runs the LDP mechanism, $\Pr[r_i=u| t_i=v]$ is the same for all voters. Hence, we denote by $p_{uv}=\Pr[r_i=u| t_i=v]$, and let  $P=(p_{uv})_{m\times m}$. $P$ is called the design matrix and is only dependent on the LDP mechanism.

We consider two major LDP mechanisms, randomized response and Laplace. Given the privacy parameter $\epsilon$, we denote by $P^{\epsilon\text{-ran}}=(p_{uv}^{\epsilon\text{-ran}})_{m\times m}$ and $P^{\epsilon\text{-lap}}=(p_{uv}^{\epsilon\text{-lap}})_{m\times m}$ the design matrices for randomized response and Laplace, respectively. According to Wang et al.~\shortcite{wang2016using}, the design matrices are given by:

\begin{equation}\label{eq:design-matrix}
p^{\epsilon\text{-ran}}_{uv}=
\begin{cases}
\theta+(1-\theta)/m & \text{if}\ u=v \\
(1-\theta)/m & \text{if}\ u\neq v \\
\end{cases} 
\end{equation}
where $\theta=1-m/(m-1+e^{\epsilon})$, and

\begin{equation}\label{eq:design-matrix-lap}
p^{\epsilon\text{-lap}}_{uv}\!=\!\!
\begin{cases}
F_{v,\frac{m-1}{\epsilon}}(\frac{3}{2}) & \text{if}\ u=1 \\
1\!-\!F_{v,\frac{m-1}{\epsilon}}(m\!-\!\frac{1}{2}) & \text{if}\ u=m \\
F_{v,\frac{m-1}{\epsilon}}(u\!+\!\frac{1}{2})\!-\!F_{v,\frac{m-1}{\epsilon}}(u\!-\!\frac{1}{2}) & \text{otherwise}  \\
\end{cases}
\end{equation}
where $F_{\mu,b}(x)=\frac{1}{2}+\frac{1}{2}sgn(x-\mu)(1-e^{(-\frac{|x-\mu|}{b})})$ denotes the cumulative distribution function of Laplace distribution with mean $\mu$ and the variance $2b^2$.

\smallskip
\noindent\textbf{From Reported Type to True Type:} Consider an arbitrary voter $V_i$. Again, let $t_i$ be its true type, and $r_i$ be its reported type. If we observe that $V_i$ reports $v$ as its type, what is the probability that its true type is $u$? More precisely, we consider $\Pr[t_i=u| r_i=v]$.
As each voter independently and locally runs the LDP mechanism, $\Pr[t_i=u| r_i=v]$ is the same for all voters. Hence, we denote by $q_{uv}=\Pr[t_i=u| r_i=v]$, and let  $Q=(q_{uv})_{m\times m}$.

The probabilities $p_{uv}$ and $q_{uv}$ are connected through Bayesian formulation. Let $\lambda_j$ be the fraction of voters whose true type is $j$ (i.e., there are $n\lambda_j$ voters whose true type is $j$). We assume $\lambda_j$'s are fixed positive values and denote $\vep_i=(p_{i1},p_{i2},\dots,p_{im})$ as the $i$-th row of the design matrix $P$. Then, we have
\begin{equation}\label{calc:quv}
    q_{uv}=\frac{\Pr[r_i=v, t_i=u]}{\sum_{u'} \Pr[r_i=v | t_i=u']\Pr[t_i=u']}=\frac{p_{vu}\lambda_u}{\boldsymbol{\lambda}\cdot{\vep_v}}.
\end{equation}

It is easy to see that $p_{uv}$'s and $q_{uv}$'s are independent of the number of voters $n$.



\section{Formal Definition of PoLDP}


The goal of this paper is to quantitatively analyze the effect of LDP mechanisms when applied to voting systems. Towards this, we consider two scenarios, the classical scenario without LDP and the new scenario with LDP, and compare the manipulation costs of the attacker under these two scenarios. Below we give details.  

We denote $\mathcal{S}_{m,\boldsymbol\lambda}(n)$ as the voting system which contains $n$ voters namely $\{V_1,\dots,V_n\}$, and $m$ candidates namely $\{C_1,\dots,C_m\}$. The parameter $\boldsymbol{\lambda}=(\lambda_1,\dots,\lambda_m)$ which denotes there are $n\lambda_i$ voters whose true type is $i$. The deletion cost for each voter is unit. Hence, there is no need to distinguish two voters who have the same true type. 

There is an attacker who tries to make a designated candidate win via voter deletion. Without loss of generality, we assume that the designated candidate is $C_1$.

In the classical scenario without LDP mechanisms, the attacker knows the true type of voters\footnote{This is a common assumption in computational social choice, see, e.g., Brandt et al.~\shortcite{brandt2016handbook}.} (e.g., the parameter $\boldsymbol{\lambda}$ of the voting system), and can thus strategically delete voters (e.g., decide the number of voters that need to be deleted for each kind of true type) to make the designated candidate $C_1$ be one of the co-winners. Hence, for the voting system $\mathcal{S}=\mathcal{S}_{m,\boldsymbol\lambda}(n)$, we define the minimal number of voters the attacker should delete as the {\it manipulation cost} of the attacker, and define it as $f(\mathcal{S})$.

\begin{figure}[!h]
	\begin{center}
		\includegraphics[width=75.66mm,height=17mm]{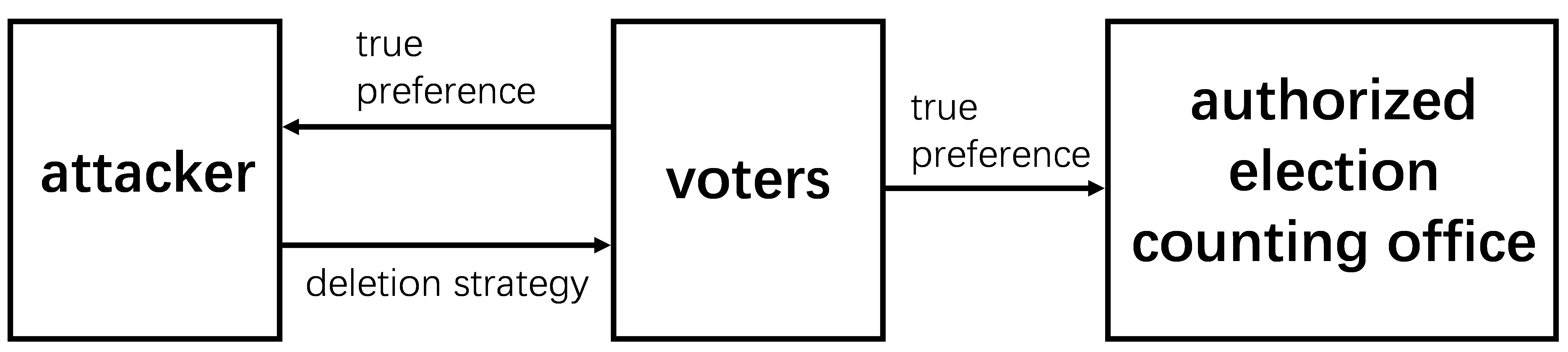}
		\caption{Illustration of Manipulation without LDP}
		\label{fig:1}	
	\end{center}
\end{figure}
\begin{figure}[!h]
	\begin{center}
		\includegraphics[width=81.77mm,height=17mm]{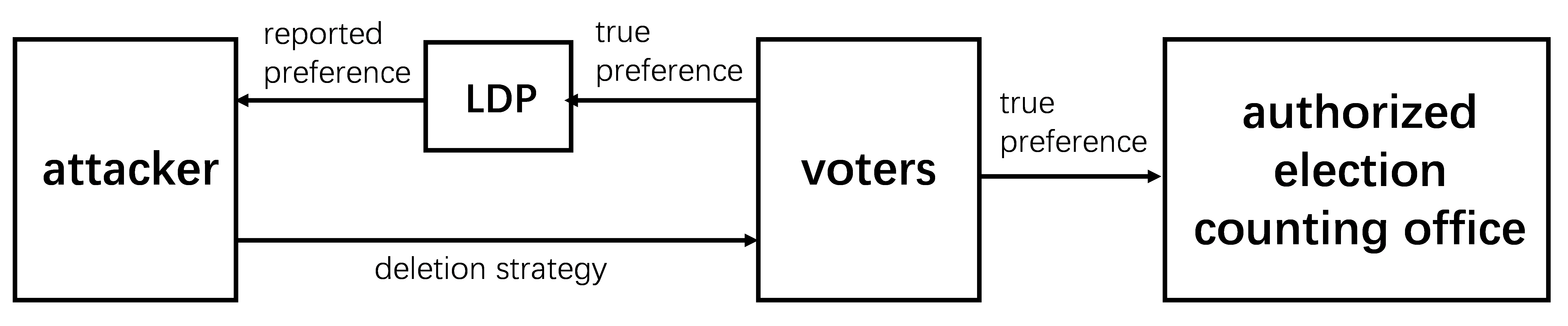}
		\caption{Illustration of Manipulation with LDP}
		\label{fig:2}	
	\end{center}
\end{figure}

Consider the scenario where an LDP mechanism $\mathcal{R}$ is adopted. We also assume the attacker knows the parameter $\boldsymbol{\lambda}$ of the voting system\footnote{For the rationality of this assumption, please refer to Section ``Discussion''.}. But, the attacker needs to delete voters according to reported types instead of true types. Adopting the LDP mechanism introduces two kinds of uncertainties when considering the manipulation problem, as we discuss below. 

\smallskip
\noindent{\it $\bullet$ Realizations:} Recall that there are $n\lambda_i$ voters whose true type is $i$. We know that it is fixed for the voting system $\mathcal{S}_{m,\boldsymbol\lambda}(n)$. However, the number of voters whose reported type is $i$, may be different\footnote{To understand the effect of an LDP mechanism to a fixed voting system on average, we consider the independent application of the same LDP mechanism multiple times, and each time we may observe a different $\boldsymbol{\tau}$.}. Let $n\tau_i$ be the number of voters whose reported type is $i$ where $1\le i\le m$, and let $\boldsymbol{\tau}=(\tau_1,\dots,\tau_m)$. Then, we know that $$\boldsymbol{\tau}\in \mathcal{P}:=\{\vez=(z_1,\cdots,z_m)\in [0,1]^m:\sum_{i=1}^m z_i=1\}.$$
For any $\vez\in\mathcal{P}$, $\boldsymbol{\tau}=\vez$ is called a {\it realization}, and will occur with the probability $\Pr[\boldsymbol{\tau}=\vez]$ that depends on the LDP mechanism $\mathcal{R}$.

\smallskip
\noindent{\it $\bullet$ Winning Probability:} 
The attacker observes a realization $\boldsymbol{\tau}$ and decides the number of voters that need to be deleted for each reported type. More precisely, a feasible solution to the attacker under a fixed realization $\boldsymbol{\tau}$ can be represented as a vector $\vex=(x_1,x_2,\cdots,x_{m})$ where $x_i\le n\tau_i$ is the number of deleted voters whose reported type is $i$. As voters of the same reported type are indistinguishable from the attacker, we assume the attacker will randomly delete $x_i$ out of $n\tau_i$ voters of reported type $i$. Among these $x_i$ voters, the number of voters whose true type is $i'$ is a stochastic value for each $i'$. Consequently, the attacker may or may not succeed in making $C_1$ win. In other words, any feasible solution $\vex$ to the attacker is associated with a {\it winning probability} which indicates the probability that $C_1$ becomes a co-winner if the attacker deletes arbitrary $x_i$ voters whose reported type is $i$, given the observation $\boldsymbol{\tau}$. 

Above all, for the voting system $\mathcal{S}_{m,\boldsymbol\lambda}(n)$, we define the manipulation cost of the attacker under the LDP mechanism as follows:

\begin{definition}
The manipulation cost of a voting system $\mathcal{S}=\mathcal{S}_{m,\boldsymbol\lambda}(n)$ under an LDP mechanism $\mathcal{R}$, given a winning probability $\xi$ and a realization $\boldsymbol{\tau}$, is the minimal number of voters that need to be deleted to make the designated candidate one co-winner with a probability at least $\xi$, and is denoted by $f(\mathcal{S},\mathcal{R},\xi:\boldsymbol{\tau})$. 
\end{definition}

As the mechanism $\mathcal{R}$ defines a distribution of $\boldsymbol{\tau}$ over $\mathcal{P}$, we can define the expectation of $f(\mathcal{S},\mathcal{R},\xi:\boldsymbol{\tau})$ over all the realizations as below.

\begin{definition}
The expected manipulation cost of a voting system $\mathcal{S}=\mathcal{S}_{m,\boldsymbol\lambda}(n)$ under an LDP mechanism $\mathcal{R}$, given a winning probability $\xi$, is the expected minimal number of voters that need to be deleted to make the designated candidate one co-winner with a probability at least $\xi$ over all the realizations, and is denoted by $f(\mathcal{S},\mathcal{R},\xi)=\mathbb{E}_{\boldsymbol{\tau}}[f(\mathcal{S},\mathcal{R},\xi:\boldsymbol{\tau})]$. 
\end{definition}

Compare the manipulation cost of the attacker in the above two scenarios (with or without LDP), the introduction of LDP may cause the cost to increase or decrease, and we measure such an increase or decrease through PoLDP, interpreted as the power of LDP. More precisely, 
\begin{definition}
The PoLDP for a voting system $\mathcal{S}=\mathcal{S}_{m,\boldsymbol\lambda}(n)$ under an LDP mechanism $\mathcal{R}$ is defined as $$\text{PoLDP}(\mathcal{S},\mathcal{R},\xi)=f(\mathcal{S},\mathcal{R},\xi)/f(\mathcal{S}).$$
\end{definition}

We are interested in the value of PoLDP when the number of voters in the voting system is sufficiently large. Hence, we define the APoLDP (Asymptotic PoLDP).

\begin{definition}\label{def:poldp-s}
The APoLDP for a voting system $\mathcal{S}=\mathcal{S}_{m,\boldsymbol\lambda}(n)$ under an LDP mechanism $\mathcal{R}$ is defined as $$\text{APoLDP}(\mathcal{S},\mathcal{R},\xi)=\lim\limits_{n\to \infty}f(\mathcal{S},\mathcal{R},\xi)/f(\mathcal{S}).$$
\end{definition}

\section{Computing APoLDP}

The definition of PoLDP (or APoLDP) relies on computing $f(\mathcal{S},\mathcal{R},\xi)$, and $f(\mathcal{S})$. In the classic scenario without LDP mechanisms, the manipulation cost of the attacker is given by $$f(\mathcal{S})=\sum_{j=2}^m\max\{0,n\lambda_j-n\lambda_1\}.$$

In scenario where an LDP mechanism $\mathcal{R}$ is adopted, the manipulation cost of the attacker $f(\mathcal{S},\mathcal{R},\xi)$ relies on computing $f(\mathcal{S},\mathcal{R},\xi:\boldsymbol{\tau})$, which is a random variable whose distribution is very complicated. We prove that under certain conditions, the value of APoLDP is robust to the winning probability $\xi$, and can be calculated via a linear program. Specifically, we have the following theorem:

\begin{theorem}\label{thm:main-new}
If the voting system $\mathcal{S}=\mathcal{S}_{m,\boldsymbol\lambda}(n)$ under an LDP mechanism $\mathcal{R}$ satisfies that for any $1\le j\le m$ 
either $q_{ij}-q_{1j}>0$ for all $2\le i\le m$, or $q_{ij}-q_{1j}<0$ for all $2\le i\le m$, then for any $\xi\in(0,1)$ it holds that 
$$\text{APoLDP}(\mathcal{S},\mathcal{R},\xi)=\frac{\text{OPT}_{\text{LP}(\boldsymbol{\hat{\tau}})}}{\sum_{j=2}^m\max\{0,n\lambda_j-n\lambda_1\}},$$
where $\hat{\tau}_i=\mathbb{E}[\tau_i]=\sum_{j=1}^m \lambda_jp_{ij}.$
\end{theorem}

The proof of Theorem~\ref{thm:main-new} is a little bit complex, we detail it in the following sections:
\begin{itemize}
    \item In subsection ``Upper and Lower Bounding $f(\mathcal{S},\mathcal{R},\xi:\boldsymbol{\tau})$ through ILP'', we will establish two integer linear programs $\overline{\text{ILP}}(\boldsymbol{\tau},\delta)$ and $\underline{\text{ILP}}(\boldsymbol{\tau},\delta)$ which gives the upper and lower bounds of $f(\mathcal{S},\mathcal{R},\xi:\boldsymbol{\tau})$.
    
    \item In subsection ``Robustness – a Sufficient Condition", we show that if certain condition holds, then the value of $f(\mathcal{S},\mathcal{R},\xi:\boldsymbol{\tau})$ is robust to the winning probability $\xi$. Specifically, we prove that for any $\epsilon\in(0,1)$, the value of $f(\mathcal{S},\mathcal{R},1-\epsilon:\boldsymbol{\tau})$ can be calculated via linear program $\text{LP}(\boldsymbol{\tau})$ when the number of voters is sufficiently large.
    
    \item In subsection ``Central Realization", we show a efficient way to calculate $f(\mathcal{S},\mathcal{R},\xi)$. Recall the definition, we need to compute the optimal objective value of linear program $\text{LP}(\boldsymbol{\tau})$ for each fixed realization $\boldsymbol{\tau}$, and then taking the expectation over $\boldsymbol{\tau}$. We prove that if certain condition holds, we can first compute the  expectation of $\boldsymbol{\tau}$ (denoted as $\hat{\boldsymbol{\tau}}$), and then compute the optimal objective value of integer program $\text{LP}(\boldsymbol{\hat{\boldsymbol{\tau}}})$.
\end{itemize}

\subsection{Upper and Lower Bounding $f(\mathcal{S},\mathcal{R},\xi:\boldsymbol{\tau})$ through ILP}
Given a fixed value $\xi$ and observation $\boldsymbol{\tau}$, the computation of $f(\mathcal{S},\mathcal{R},\xi:\boldsymbol{\tau})$ is very complicated, which makes it very difficult to determine whether $\mathcal{S}$ is stable or not. Towards this, we will introduce two integer linear programs (ILPs) which can be used to compute an upper bound and lower bound for $f(\mathcal{S},\mathcal{R},\xi:\boldsymbol{\tau})$, respectively. 



Let $x_j$ be a variable that denotes the number of voters with reported type $j$ that need to be deleted. Recall that for any voter, $q_{ij}$ denotes the probability that its true type is $i$ conditioned on that its reported type is $j$. Consequently, if we delete $x_j$ voters with reported type $j$, then in expectation we delete $q_{ij}x_j$ voters with true type $i$. Denote by $\hat{\Gamma}_i$ the expected number of voters with true type $i$ after voter deletion. 

Our first observation is that, if $\hat{\Gamma}_1-\hat{\Gamma}_i$ is sufficiently large, then the designated candidate shall win with sufficiently high probability. Specifically, we consider the following  $\overline{\text{ILP}}(\boldsymbol{\tau},\delta)$: \begin{subequations}
\begin{eqnarray}
&\hspace{-15mm}\min&\hspace{-5mm} \sum_{i=1}^m x_i \nonumber\\
&\hspace{-15mm}s.t.&  \hspace{-5mm} \hat{\Gamma}_i=n\lambda_i-\sum_{j=1}^m q_{ij}x_j \qquad\qquad\; \forall\;
1\le i \le m \label{uip:1}\\
&& \hspace{-5mm} \hat{\Gamma}_1\geq \hat{\Gamma}_i +  \delta\sum_{j=1}^m (q_{1j}+q_{ij})x_j  \quad \forall\; 2\le i\le m \label{uip:2}\\
&& \hspace{-5mm} x_i\leq n\tau_i  \qquad\qquad\qquad\qquad\quad\;\, \forall\;
1\le i \le m \label{uip:3}\\
&& \hspace{-5mm} x_i\in \mathbb{N}  \nonumber
\end{eqnarray}
\end{subequations}


The optimal objective value of $\overline{\text{ILP}}(\boldsymbol{\tau},\delta)$ gives an upper bound on $f(\mathcal{S},\mathcal{R},\xi:\boldsymbol{\tau})$, as indicated by the following lemma.
\begin{lemma}\label{lemma:upper}
Let $\vex^*=(x_1^*,\dots,x_m^*)$ denote the optimal solution to $\overline{\text{ILP}}(\boldsymbol{\tau},\delta)$. If the attacker deletes arbitrary $x_i^*$ voters whose reported type is $i$, then the winning probability of the designated candidate is at least $1-me^{-\frac{c\delta^2}{3}n}$, where $c= q_{\min}(\lambda_{\max}-\lambda_1)$,$q_{\min}=\min_{i,j}q_{ij}$, $\lambda_{\max}=\max_j\lambda_j$.
\end{lemma}
Note that according to the definition of $q_{ij}$'s, $q_{ij}>0$ if $p_{ji}>0$, which is satisfied by randomized response and Laplace mechanisms. 

\begin{proof}
Let the set of deleted voters be $D$. Let $\Gamma_i$ be the number of voters whose true type is $i$ after deletion, which is a random variable. 

For any $V_j\in D$, consider the event $EV_{ij}$ that the true type of voter $V_j$ is $i$. Define the binary random variable $L_{ij}$ such that $L_{ij}=1$ if $EV_{ij}$ occurs, and $L_{ij}=0$ otherwise. We know that $n\lambda_i-\Gamma_i$ can be expressed as: 
$$n\lambda_i-\Gamma_i=\sum_{j\in D} L_{ij}.$$

Hence,
$$\mathbb{E}[n\lambda_i-\Gamma_i]=\sum_{j\in D} \Pr[L_{ij}=1]=\sum_{j=1}^m x_j^* q_{ij}.$$

We claim that $\sum_jx_j^*\ge (\lambda_{\max}-\lambda_1)n$. To see this, let $\lambda_k=\lambda_{\max}$. Given that $\hat{\Gamma_1}-\hat{\Gamma}_k\ge 0$, we know $\sum_{j}(q_{kj}-q_{1j})x_j^*\ge n(\lambda_k-\lambda_1)$. As $|q_{kj}-q_{1j}|\le 1$, it follows directly that $\sum_{j}(q_{kj}-q_{1j})x_j^*\ge \sum_{j}(q_{kj}-q_{1j})x_j^*\ge n(\lambda_k-\lambda_1)$. Hence, the claim is true.

Given the above claim, it is easy to see that $$\mathbb{E}[n\lambda_i-\Gamma_i] \geq q_{\min}\sum_{j=1}^n x_j^* \geq q_{\min}(\lambda_{\max}-\lambda_1)n=cn.$$


Consider the probability that the designated candidate $C_1$ becomes one co-winner after voter deletion, which can be expressed as $\prod_{i=2}^m \Pr[\Gamma_1 \geq \Gamma_i]$.
Given that $\vex^*$ guarantees that $\hat{\Gamma}_1-\hat{\Gamma}_{i}\ge \delta\sum_{j=1}^m (q_{1j}+q_{ij})x_j^*$, if the event occurs that $\Gamma_1\geq \hat{\Gamma}_1 -\delta\sum_{j=1}^m q_{1j}x_j^*$, and meanwhile $\Gamma_i\leq \hat{\Gamma}_i +\delta\sum_{j=1}^m q_{ij}x_j^*$, then it holds directly that $\Gamma_1\ge \Gamma_i$. Hence,  
\begin{eqnarray*}
&&\prod_{i=2}^m \Pr[\Gamma_1 \geq \Gamma_i] \\
&\geq& \Pr[\Gamma_1\geq \hat{\Gamma}_1 -\delta\sum_{j=1}^m q_{1j}x_j^*] \prod_{i=2}^m \Pr[\Gamma_i\leq \hat{\Gamma}_i+\delta\sum_{j=1}^m q_{ij}x_j^*].
\end{eqnarray*}

By Chernoff Bounds, it is easy to see that $$\Pr[\Gamma_1\geq \hat{\Gamma}_1-\delta\sum_{j=1}^m q_{1j}x_j^*] \geq 1 - e ^{-\frac{c\delta^2}{3}n},$$ $$\Pr[\Gamma_i\leq \hat{\Gamma}_i+\delta\sum_{j=1}^m q_{ij}x_j^*] \geq 1 - e ^{-\frac{c\delta^2}{3}n}.$$

Hence, $$\prod_{i=2}^m \Pr[\Gamma_1 \geq \Gamma_i] \ge (1 - e ^{-\frac{c\delta^2}{3}n})^m\geq 1 - me^{-\frac{c\delta^2}{3}n}. \qedhere$$
\end{proof}

Similarly, we consider another integer linear program $\underline{\text{ILP}}(\boldsymbol{\tau},\delta)$ as follow \begin{subequations}
\begin{eqnarray}
&\hspace{-15mm}\min&\hspace{-5mm} \sum_{i=1}^m x_i \nonumber\\
&\hspace{-15mm}s.t.&  \hspace{-5mm} \hat{\Gamma}_i=n\lambda_i-\sum_{j=1}^m q_{ij}x_j \qquad\qquad\; \forall\;
1\le i \le m \label{lip:1}\\
&& \hspace{-5mm} \hat{\Gamma}_1\geq \hat{\Gamma}_i -  \delta\sum_{j=1}^m (q_{1j}+q_{ij})x_j  \quad \forall\; 2\le i\le m \label{lip:2}\\
&& \hspace{-5mm} x_i\leq n\tau_i  \qquad\qquad\qquad\qquad\quad\;\, \forall\;
1\le i \le m \label{lip:3}\\
&& \hspace{-5mm} x_i\in \mathbb{N}  \nonumber
\end{eqnarray}
\end{subequations}
where again $x_j$ denote the number of voters with reported type $j$ and will be deleted.


\begin{lemma}\label{lemma:lower}
Let $\vex^*=(x_1^*,\dots,x_m^*)$ denote the optimal solution to $\underline{\text{ILP}}(\boldsymbol{\tau},\delta)$.
If the attacker deletes $y_i$ voters whose reported type is $i$ such that 
$\sum_jy_j<\sum_jx_j^*:=\text{OPT}_{\underline{\text{ILP}}(\boldsymbol{\tau},\delta)}$, then the winning probability of the designated candidate is less than $2e^{-\frac{c\delta^2}{3}n}$ where $c= q_{\min}(\lambda_{\max}-\lambda_1)$, $q_{\min}=\min_{i,j}q_{ij}$, $\lambda_{\max}=\max_j\lambda_j$. 
\end{lemma}

\begin{proof}
Define $\hat{\Gamma}_i(\vey)=n\lambda_i-\sum_{j=1}^m q_{ij}y_j$.
Given that $\vey=(y_1,\cdots,y_m)$ is not a feasible solution to $\underline{\text{ILP}}(\boldsymbol{\tau},\delta)$ (as its objective value is smaller than the optimal objective value), it must hold that for some $2\le i^*\le m$, 
$$\hat{\Gamma}_1(\vey)<\hat{\Gamma}_{i^*}(\vey)-\delta\sum_{j=1}^m (q_{1j}+q_{i^*j})y_j.$$

Again Let $\Gamma_i(\vey)$ be the number of voters whose true type is $i$ after the attacker deletes $y_i$ voters with reported type $i$, which is a random variable. Using the same argument as Lemma~\ref{lemma:upper}, we know that
\begin{eqnarray*}
&&\Pr[\Gamma_1(\vey)\leq \hat{\Gamma}_1(\vey)+\delta\sum_{j=1}^m q_{1j}y_j] \geq 1 - e ^{-\frac{c\delta^2}{3}n},\\
&&\Pr[\Gamma_{i^*}(\vey)\geq \hat{\Gamma}_{i^*}(\vey)-\delta\sum_{j=1}^m q_{i^*j}y_j] \geq 1 - e ^{-\frac{c\delta^2}{3}n},
\end{eqnarray*}
where $c= q_{\min}(\lambda_{\max}-\lambda_1)$.

Consequently, 
\begin{eqnarray*}
&& \Pr[\Gamma_1(\vey) < \Gamma_{i^*}(\vey)] \\
&\geq& \Pr[\Gamma_1(\vey)\leq \hat{\Gamma}_1(\vey) +\delta\sum_{j=1}^m q_{1j}y_j] \\&\times& \Pr[\Gamma_{i^*}(\vey)\geq \hat{\Gamma}_{i^*}(\vey)-\delta\sum_{j=1}^m q_{i^*j}y_j]\\
&\ge& (1 - e ^{-\frac{c\delta^2}{3}n})^2\ge 1 - 2e ^{-\frac{c\delta^2}{3}n}
\end{eqnarray*}
That is, the winning probability of the attacker is at most $2e ^{-\frac{c\delta^2}{3}n}$, and Lemma~\ref{lemma:lower} is proved.
\end{proof}

\subsection{Robustness---a Sufficient Condition}
In previous section, we have established two integer linear programs to characterize the lower and upper bound of $f(\mathcal{S},\mathcal{R},\xi:\boldsymbol{\tau})$ in general. In this section, we will prove that under some mild assumptions, the lower and upper bounds will approach the same value when the number of voters is sufficiently large.

We introduce the following ${\text{ILP}}(\boldsymbol{\tau})$ as an \lq\lq intermediate\rq\rq\, integer linear program between
$\overline{\text{ILP}}(\boldsymbol{\tau},\delta)$ and $\underline{\text{ILP}}(\boldsymbol{\tau},\delta)$: \begin{subequations}
\begin{eqnarray}
&\hspace{-15mm}\min&\hspace{-5mm} \sum_{i=1}^m x_i \nonumber\\
&\hspace{-15mm}s.t.&  \hspace{-5mm} \hat{\Gamma}_i=n\lambda_i-\sum_{j=1}^m q_{ij}x_j \qquad\qquad \forall\;
1\le i \le m \label{mip:1}\\
&& \hspace{-5mm} \hat{\Gamma}_1\geq \hat{\Gamma}_i  \qquad\qquad\qquad\qquad\quad\;\, \forall\; 2\le i\le m \label{mip:2}\\
&& \hspace{-5mm} x_i\leq n\tau_i  \qquad\qquad\qquad\qquad\quad\, \forall\;
1\le i \le m \label{mip:3}\\
&& \hspace{-5mm} x_i\in \mathbb{N}  \nonumber
\end{eqnarray}
\end{subequations}

\noindent\textbf{Remark on the integral constraint.} We denote ${\text{LP}}(\boldsymbol{\tau})$ as the linear program which ignores the integral constraint of ${\text{ILP}}(\boldsymbol{\tau})$. There are $O(m)$ constraints in ${\text{ILP}}(\boldsymbol{\tau})$. Hence, the additive gap between ${\text{ILP}}(\boldsymbol{\tau})$ and its linear relaxation ${\text{LP}}(\boldsymbol{\tau})$ is bounded by $O(m)$. Meanwhile, we know that there exists constant $c$ such that $\text{OPT}_{{\text{ILP}}(\boldsymbol{\tau})}\geq cn$. It turns out when $n$ is sufficiently large, we can ignore the additive gap between ${\text{ILP}}(\boldsymbol{\tau})$ and its linear relaxation. In other words, it holds that $\lim \limits_{n\to \infty} \text{OPT}_{{\text{LP}}(\boldsymbol{\tau})} / \text{OPT}_{{\text{ILP}}(\boldsymbol{\tau})}=1$. For integer linear program $\overline{\text{ILP}}(\boldsymbol{\tau},\delta)$ and $\overline{\text{ILP}}(\boldsymbol{\tau},\delta)$ introduced in previous section, we can conclude similar results (e.g., $\lim \limits_{n\to \infty} \overline{\text{LP}}(\boldsymbol{\tau},\delta) / \overline{\text{ILP}}(\boldsymbol{\tau},\delta)=1$, $\lim \limits_{n\to \infty} \underline{\text{LP}}(\boldsymbol{\tau},\delta) / \underline{\text{ILP}}(\boldsymbol{\tau},\delta)=1$ where $\overline{\text{LP}}(\boldsymbol{\tau},\delta)$ and $\underline{\text{LP}}(\boldsymbol{\tau},\delta)$ represent the linear relaxation respectively).

According to Lemma~\ref{lemma:upper} and Lemma~\ref{lemma:lower}, taking $\delta_0=n^{-1/3}$ and $\xi_0=me^{-\frac{c\delta_0^2}{3}n}\ge 2e^{-\frac{c}{3}n^{1/3}}$, we know 
\begin{eqnarray}
f(\mathcal{S},\mathcal{R},1-\xi_0:\boldsymbol{\tau})\le \text{OPT}_{\overline{\text{ILP}}(\boldsymbol{\tau},\delta_0)}, \label{eq:1}
\end{eqnarray} and 
\begin{eqnarray}
f(\mathcal{S},\mathcal{R},\xi_0:\boldsymbol{\tau})\ge \text{OPT}_{\underline{\text{ILP}}(\boldsymbol{\tau},\delta_0)}. \label{eq:2}
\end{eqnarray}

\begin{lemma}\label{lemma:stable-condition}
For any $1\le j\le m$ 
if $q_{ij}-q_{1j}>0$ for all $2\le i\le m$, or $q_{ij}-q_{1j}<0$ for all $2\le i\le m$, 
then there exists some $c'>0$ which is independent of $n$, such that for sufficiently small $\delta$ the followings hold:
\begin{eqnarray}
\text{OPT}_{{\text{LP}}(\boldsymbol{\tau})}\le \text{OPT}_{\overline{\text{LP}}(\boldsymbol{\tau},\delta)}\le (1+c'\delta)\text{OPT}_{{\text{LP}}(\boldsymbol{\tau})}, \label{eq:ILP-m1}
\end{eqnarray}
and 
\begin{eqnarray}
 \text{OPT}_{\underline{\text{LP}}(\boldsymbol{\tau},\delta)}\le \text{OPT}_{{\text{LP}}(\boldsymbol{\tau})}\le  (1+c'\delta)\text{OPT}_{\underline{\text{LP}}(\boldsymbol{\tau},\delta)}. \label{eq:ILP-m2}
\end{eqnarray}
\end{lemma}

By Lemma~\ref{lemma:stable-condition}, we know that $$\frac{\text{OPT}_{\overline{\text{LP}}(\boldsymbol{\tau},\delta_0)}}{\text{OPT}_{\underline{\text{LP}}(\boldsymbol{\tau},\delta_0)}}\le (1+c'\delta_0)^2.$$

Combine all, we know that for any $\xi\in(0,1)$, it holds that $$\lim \limits_{n \to \infty} f(\mathcal{S},\mathcal{R},\xi:\boldsymbol{\tau}) = \text{OPT}_{{\text{LP}}(\boldsymbol{\tau})} $$

The following part of this section is devoted to the proof of Lemma~\ref{lemma:stable-condition}.

\begin{proof}[Proof of Lemma~\ref{lemma:stable-condition}]
We prove Eq~\eqref{eq:ILP-m1} in the following. The proof of Eq~\eqref{eq:ILP-m2} follows from a similar argument.

Notice that for all $2\le i\le m$, we have
\begin{eqnarray}
\hat{\Gamma}_1-\hat{\Gamma}_i=n(\lambda_1-\lambda_i)+\sum_{j=1}^m(q_{ij}-q_{1j})x_{j}. \label{eq:gamma}
\end{eqnarray}
Let $\bar{\vex}^*$, ${\vex}^*$ and $\underline{\vex}^*$ be the optimal solutions to $\overline{\text{LP}}(\boldsymbol{\tau},\delta)$, ${\text{LP}}(\boldsymbol{\tau})$ and $\underline{\text{LP}}(\boldsymbol{\tau},\delta)$ respectively.

We first observe that if for some $j=j_0$, $q_{ij_0}-q_{1j_0}<0$ holds for all $2\le i\le m$, then $x_{j_0}^*=0$. This is because the coefficient of $x_{j_0}^*$ is negative for all $i$, if $x_{j_0}^*>0$, reducing it to $0$ the constraint of $\hat{\Gamma}_1-\hat{\Gamma}_i\ge 0$ still holds, while the objective value $\sum_ix^*_i$ decreases, contradicting to the optimality of $\vex^*$. Similarly, when $\delta$ is sufficiently small, we also have $(1-\delta)q_{ij_0}-(1+\delta)q_{1j_0}<0$ and $(1+\delta)q_{ij_0}-(1-\delta)q_{1j_0}<0$. Using the same argument we can conclude that $\bar{x}_{j_0}^*=0$ and $\underline{x}_{j_0}^*=0$. Hence, we can simply remove the variable $x_{j_0}$ for all such $j_0$'s. In the following we assume $q_{ij}-q_{1j}>0$ holds for all $2\le i\le m$ and $1\le j\le m$. 

We first compare ${\text{LP}}(\boldsymbol{\tau})$ with $\overline{\text{LP}}(\boldsymbol{\tau},\delta)$. 
By Eq~\eqref{uip:2} we know
$$\sum_{j=1}^m\left(q_{ij}(1-\delta)-q_{1j}(1+\delta) \right)\bar{x}^*_j\ge n(\lambda_i-\lambda_1).$$
Consequently, it follows that 
$$(1-\delta)\sum_{j=1}^m\left(q_{ij}-q_{1j} \right)\bar{x}^*_j\ge n(\lambda_i-\lambda_1),$$
implying that $(1-\delta)\bar{\vex}^*$ is a feasible solution to ${\text{LP}}(\boldsymbol{\tau})$. Thus
\begin{eqnarray*}
\text{OPT}_{{\text{LP}}(\boldsymbol{\tau})}\le (1-\delta)\text{OPT}_{\overline{\text{LP}}(\boldsymbol{\tau},\delta)}. \label{eq:medium-1}
\end{eqnarray*}

Next, we show that $\text{OPT}_{\overline{\text{LP}}(\boldsymbol{\tau},\delta)}$ is not too large compared with $\text{OPT}_{{\text{LP}}(\boldsymbol{\tau})}$. Towards this, define
\begin{eqnarray}
c_0=\frac{\max_{1\le j\le m}\max_{2\le i\le m}\{q_{ij}-q_{1j}\}}{\min_{1\le j\le m}\min_{2\le i\le m}\{q_{ij}-q_{1j}\}}, \label{eq:c_0}
\end{eqnarray}
and 
\begin{eqnarray}
c_1=\max_{1\le j\le m}\max_{2\le i\le m}\frac{q_{ij}+q_{1j}}{q_{ij}-q_{1j}}. \label{eq:c_1}
\end{eqnarray}
Notice that $\text{OPT}_{\overline{\text{LP}}(\boldsymbol{\tau},\delta)}\le \sum_in\tau_i=n$. 

If $\text{OPT}_{{\text{LP}}(\boldsymbol{\tau})}>(1-c_0\varepsilon)n$, then it follows directly that 
\begin{eqnarray*}
\text{OPT}_{\overline{\text{LP}}(\boldsymbol{\tau},\delta)}\le n &\le& \frac{1}{1-c_0c_1\delta}\text{OPT}_{{\text{LP}}(\boldsymbol{\tau})}\\&\le& \frac{1}{1-c_1\delta}(1+\frac{c_0c_1\delta}{1-c_1\delta})\text{OPT}_{{\text{LP}}(\boldsymbol{\tau})}.
\end{eqnarray*}

Otherwise, $\text{OPT}_{{\text{LP}}(\boldsymbol{\tau})}\le(1-c_0c_1\delta)n$. We make the following claims:
\begin{claim}\label{claim:1}
If $\text{OPT}_{{\text{LP}}(\boldsymbol{\tau})}\le(1-c_0c_1\delta)n$, then there exists a near optimal solution $\vex'$ to ${\text{LP}}(\boldsymbol{\tau})$ such that $x'_i\le n\tau_i(1-c_1\delta)$, and $\sum_{i}x_i'\le (1+\frac{c_0c_1\delta}{1-c_1\delta}) \text{OPT}_{{\text{LP}}(\boldsymbol{\tau})}$.
\end{claim}
\begin{proof}[Proof of Claim~\ref{claim:1}] 
We construct $\vex'$ by modifying $\vex^*$. Let $I=\{1\le i\le m: x_i^*\ge n\tau_i(1-c_1\delta)\}$. For $i\in I$, we let $\Delta_i=x_i^*-n\tau_i(1-c_1\delta)\ge 0$. For $i\not\in I$, we let  $\Delta_i=n\tau_i(1-c_1\delta)-x_i^*\ge 0$. It is easy to see that
\begin{eqnarray*}
\sum_{i\not\in I} \Delta_i-\sum_{i\in I} \Delta_i&=&\sum_i n\tau_i(1-c_1\delta)-\sum_ix_i^*\\
&\ge& n(1-c_1\delta)-n(1-c_0c_1\delta)\\
&\ge& (c_0c_1-c_1)\delta n
\end{eqnarray*}
Meanwhile, $\sum_{i\in I} \Delta_i\le \sum_{i\in I} c_1\delta n\tau_i\le c_1\delta n$, hence,
$$\sum_{i\not\in I} \Delta_i-c_0\sum_{i\in I} \Delta_i\ge 0.$$
For $i\not\in I$, pick arbitrary $0\le\Delta_i'\le \Delta_i$ such that 
$$\sum_{i\not\in I} \Delta_i'-c_0\sum_{i\in I} \Delta_i= 0.$$

We define $\vex'$ as follows.
For $i\in I$, we let $x'_i=n\tau_i(1-c_1\delta)=x^*_i-\Delta_i$; For $i\not\in I$, we let $x'_i=x_i+\Delta_i$. 
It is easy to see that $x_i'\le n\tau_i(1-c_1\delta)$ for all $i$. 

We show that $\vex'$ constructed above also satisfies Eq~\eqref{mip:1} and Eq~\eqref{mip:2}, and is thus feasible to ${\text{LP}}(\boldsymbol{\tau})$. Recall Eq~\eqref{eq:gamma}, we know $\sum_{j=1}^m (q_{ij}-q_{1j})x_j^*\ge n(\lambda_i-\lambda_1)$. Meanwhile, we have that
\begin{eqnarray*}
&&\sum_{j=1}^m (q_{ij}-q_{1j})x_j'\\
&=&\sum_{j\in I} (q_{ij}-q_{1j})(x_j^*-\Delta_i)+\sum_{j\in I} (q_{ij}-q_{1j})(x_j^*+\Delta_i)\\
&=&\sum_{j=1}^m (q_{ij}-q_{1j})x_j^*+\sum_{j\not\in I}(q_{ij}-q_{1j})\Delta_j-\sum_{j\in I}(q_{ij}-q_{1j})\Delta_j\\
&\ge& n(\lambda_i-\lambda_1)+\min_{i,j}\{q_{ij}-q_{1j}\}\sum_{j\not\in I}\Delta_j\\&-&\max_{i,j}\{q_{ij}-q_{1j}\}\sum_{j\in I}\Delta_j\\
&\ge& n(\lambda_i-\lambda_1)+\min_{i,j}\{q_{ij}-q_{1j}\}(\sum_{j\not\in I}\Delta_j-c_0\sum_{j\in I}\Delta_j)\\
&\ge& n(\lambda_i-\lambda_1)
\end{eqnarray*}
Hence, $\vex'$ is feasible to ${\text{LP}}(\boldsymbol{\tau})$.

Finally, we compare the objective value of $\vex'$ and $\vex^*$. We observe that
$$\text{OPT}_{{\text{LP}}(\boldsymbol{\tau})}=\sum_ix_i^*\ge \sum_{i\in I}x_i^*\ge (1-c_1\delta)\sum_{i\in I}n\tau_i.$$
Meanwhile,
\begin{eqnarray*}
\sum_{i}x_i'-\sum_ix_i^*&\le& \sum_{i\not\in I} \Delta_i'-\sum_{i\in I}\Delta_i\\
&\le& c_0 \sum_{i\in I}\Delta_i\\
&\le & c_0c_1\delta\sum_{i\in I} n\tau_i
\end{eqnarray*}
Consequently, $\sum_{i}x_i'\le (1+\frac{c_0c_1\delta}{1-c_1\delta}) \text{OPT}_{{\text{LP}}(\boldsymbol{\tau})}$.
\end{proof}

Now we construct a feasible solution of $\overline{\text{LP}}(\boldsymbol{\tau},\delta)$ based on $\vex'$. 

\begin{claim}\label{claim:2}
If $\text{OPT}_{{\text{LP}}(\boldsymbol{\tau})}\le(1-c_0c_1\delta)n$, then let $\vex'$ be the solution to ${\text{LP}}(\boldsymbol{\tau})$ satisfying Claim~\ref{claim:1}. Then $\frac{1}{1-c_1\delta}\vex'$ is a feasible solution to $\overline{\text{LP}}(\boldsymbol{\tau},\delta)$.
\end{claim}
\begin{proof}[Proof of Claim~\ref{claim:2}]
As $\vex'$ is feasible to ${\text{LP}}(\boldsymbol{\tau})$, we know 
$$\sum_{j=1}^m(q_{ij}-q_{1j})x_{j}'\ge n(\lambda_i-\lambda_1), \forall 2\le i\le m.$$
Consequently,
\begin{eqnarray*}
&&n(\lambda_i-\lambda_1)\\
&\le &\sum_{j=1}^m(q_{ij}-q_{1j})x_{j}'\\
&=&\sum_{j=1}^m\left((1-\delta)q_{ij}-(1+\delta)q_{1j}\right)\cdot\frac{q_{ij}-q_{1j}}{(1-\delta)q_{ij}-(1+\delta)q_{1j}}x_{j}'\\
&=&\sum_{j=1}^m\left((1-\delta)q_{ij}-(1+\delta)q_{1j}\right)\cdot\frac{1}{1-\delta\frac{q_{ij}+q_{1j}}{q_{ij}-q_{1j}}}x_{j}'\\
&\le& \sum_{j=1}^m\left((1-\delta)q_{ij}-(1+\delta)q_{1j}\right)\cdot\frac{1}{1-c_1\delta}x_j'
\end{eqnarray*}
Furthermore, $\frac{1}{1-c_1\delta}x_j'\le n\tau_j$ by Claim~\ref{claim:1}. Thus, $\frac{1}{1-c_1\delta}\vex'$ is a feasible solution to $\overline{\text{LP}}(\boldsymbol{\tau},\delta)$. 
\end{proof}
Given that $\bar{\vex}^*$ is optimal to $\overline{\text{LP}}(\boldsymbol{\tau},\delta)$, by Claim~\ref{claim:2} we get that
$$\text{OPT}_{\overline{\text{LP}}(\boldsymbol{\tau},\delta)}=\sum_i\bar{x}_i^*\le \frac{1}{1-c_1\delta}\sum_{i}x_i'.$$
Meanwhile by Claim~\ref{claim:1} we have
$$\sum_{i}x_i'\le (1+\frac{c_0c_1\delta}{1-c_1\delta}) \text{OPT}_{{\text{LP}}(\boldsymbol{\tau})},$$
hence
\begin{eqnarray}
\text{OPT}_{\overline{\text{LP}}(\boldsymbol{\tau},\delta)}\le \frac{1}{1-c_1\delta}(1+\frac{c_0c_1\delta}{1-c_1\delta}) \text{OPT}_{{\text{LP}}(\boldsymbol{\tau})}. \label{eq:m2}
\end{eqnarray}
To summarize, Eq~\eqref{eq:m2} always holds no matter $\text{OPT}_{{\text{LP}}(\boldsymbol{\tau})}\le(1-c_0c_1\delta)n$ holds or not. Hence, Eq~\eqref{eq:ILP-m1} is true. Using the same argument, we can also prove Eq~\eqref{eq:ILP-m2}.
\end{proof}

\subsection{Central Realization}
In this section we deal with the realization $\boldsymbol{\tau}$. By definition $f(\mathcal{S},\mathcal{R},\xi)=\mathbb{E}_{\boldsymbol{\tau}}[f(\mathcal{S},\mathcal{R},\xi:\boldsymbol{\tau})]$ which requires us to consider all possible realizations within 
$\mathcal{P}:=\{\vez=(z_1,\cdots,z_m)\in [0,1]^m:\sum_{i=1}^m z_i=1\}$. We show that, this expectation concentrate on a single point $\boldsymbol{\tau}=\boldsymbol{\hat{\tau}}$, where $\hat{\tau}_i=\mathbb{E}[\tau_i]=\sum_{j=1}^m \lambda_jp_{ij}.$ More precisely, 


\begin{lemma}\label{thm:p-main}
Let $\hat{\tau}_i=\mathbb{E}[\tau_i]=\sum_{j=1}^m \lambda_jp_{ij}.$ For any $1\le j\le m$, 
if $q_{ij}-q_{1j}>0$ for all $2\le i\le m$, or $q_{ij}-q_{1j}<0$ for all $2\le i\le m$, then
for arbitrary $\varepsilon \geq 3 \sqrt{\frac{\log n}{n}}$, it holds that 
$$\mathbb{E}_{\boldsymbol{{\tau}}}[\text{OPT}_{\text{LP}(\boldsymbol{{\tau}})}]\in [(1-\varepsilon)\text{OPT}_{\text{LP}(\hat{\boldsymbol{{\tau}}})}, (1+\varepsilon)\text{OPT}_{\text{LP}(\hat{\boldsymbol{{\tau}}})}].$$ 
\end{lemma}

That is, the expectation of the optimal objective value of $\text{OPT}_{\text{LP}(\boldsymbol{{\tau}})}$ is arbitrarily close to its optimal objective value at $\boldsymbol{\tau}=\boldsymbol{\hat{\tau}}$. Combining Lemma~\ref{lemma:stable-condition} and Lemma~\ref{thm:p-main}, we conclude that $\lim\limits_{n\to \infty}f(\mathcal{S},\mathcal{R},\xi)=\mathbb{E}_{\boldsymbol{{\tau}}}[\text{OPT}_{\text{LP}(\boldsymbol{{\tau}})}]=\text{OPT}_{\text{LP}(\boldsymbol{\hat{\tau}})}$. Hence, Theorem~\ref{thm:main-new} is proved.

The following part of this section is devoted to the proof of Lemma~\ref{thm:p-main}. The proof relies on the following lemmas. 

\begin{lemma}\label{lemma:tau-1} Let $\hat{\tau}_i=\mathbb{E}[\tau_i]=\sum_{j=1}^n \lambda_jp_{ij}$, it holds that $$\Pr[(1-\varepsilon)\hat{\tau}_i \leq \tau_i \leq (1+\varepsilon)\hat{\tau}_i] \geq 1-e^{-\frac{5\varepsilon^2}{6}n}.$$ 
\end{lemma}

\begin{proof}[Proof of Lemma~\ref{lemma:tau-1}]
We consider the event $EV_{ij}$ that the reported type of voter $V_j$ is $i$. Define binary random variable $K_{ij}$ such that $K_{ij}=1$ if $EV_{ij}$ occurs, and $K_{ij}=0$ otherwise. Let $N_i=n\tau_i$ be the number of voters whose reported type is $i$. Note that $N_i$ is a random variable and can be expressed as $$N_i=\sum_{j=1}^n K_{ij}.$$
Thus,
$$\mathbb{E}[N_i]=\sum_{j=1}^n \Pr[K_{ij}=1]=\sum_{j=1}^m (n\lambda_j)p_{ij} \geq c n$$

By Chernoff Bounds, we know that $$\Pr[N_i\leq (1+\varepsilon)\mathbb{E}[N_i]]\geq  1-e^{-\frac{\varepsilon^2}{3}n},$$ $$\Pr[N_i\geq (1-\varepsilon)\mathbb{E}[N_i]]\geq  1-e^{-\frac{\varepsilon^2}{2}n}.$$

Hence, it holds that 
\begin{eqnarray*}
&&\Pr[(1-\varepsilon)\mathbb{E}[N_i]] \leq N_i \leq (1+\varepsilon)\mathbb{E}[N_i]]]\\
&\geq& (1-e^{-\frac{\varepsilon^2}{2}n})(1-e^{-\frac{\varepsilon^2}{3}n})\geq 1-e^{-\frac{5\varepsilon^2}{6}n}
\end{eqnarray*}

Using that $\tau_i=N_i/n$, and $\hat{\tau}_i=\mathbb{E}[\tau_i]=\mathbb{E}[N_i]/n$, Lemma~\ref{lemma:tau-1}  is true.
\end{proof}

Denote by $$\hat{\mathcal{P}}=[(1-\varepsilon)\hat{\tau}_1,(1+\varepsilon)\hat{\tau}_1]\times\dots\times[(1-\varepsilon)\hat{\tau}_m,(1+\varepsilon)\hat{\tau}_m].$$ Our next target is to show that when the observation $\boldsymbol{\tau}$ lies in the polytope $\hat{\mathcal{P}}$, $\text{OPT}_{\text{LP}(\boldsymbol{{\tau}})}$ is sufficiently close to $\text{OPT}_{\text{LP}(\hat{\boldsymbol{{\tau}}})}$. More precisely,

\begin{lemma}\label{lemma:domain-1}
For any $1\le j\le m$, 
if $q_{ij}-q_{1j}>0$ for all $2\le i\le m$, or $q_{ij}-q_{1j}<0$ for all $2\le i\le m$, 
then there exists some $\bar{c}'>0$ which is independent of $n$, such that for any $\boldsymbol{\tau}\in\hat{\mathcal{P}}$ it holds that:
$$|\text{OPT}_{\text{LP}(\boldsymbol{\tau})}-\text{OPT}_{\text{LP}(\boldsymbol{\hat{\tau}})}| \leq \bar{c}'\varepsilon \text{OPT}_{\text{LP}(\boldsymbol{\hat{\tau}})}.$$
\end{lemma}

\begin{proof}[Proof of Lemma~\ref{lemma:domain-1}]
We use a similar idea as the proof of Lemma~\ref{lemma:stable-condition}. Let $\vex^*(\boldsymbol{{\tau}})$ be the optimal solution to $\text{LP}(\boldsymbol{\tau})$. 

First, We observe that if for some $j=j_0$, $q_{ij_0}-q_{1j_0}<0$ holds for all $2\le i\le m$, then $x_{j_0}^*(\boldsymbol{\tau})=0$. This is because the coefficient of $x_{j_0}^*(\boldsymbol{\tau})$ is negative for all $i$, if $x_{j_0}^*(\boldsymbol{\tau})>0$, reducing it to $0$ the constraint $\hat{\Gamma}_1-\hat{\Gamma}_i\ge 0$ still holds, while the objective value $\sum_ix^*_i(\boldsymbol{\tau})$ decreases, contradicting to the optimality of $\vex^*(\boldsymbol{\tau})$. Hence, we can simply remove the variable $x_{j_0}$ for all such $j_0$'s. In the following we assume $q_{ij}-q_{1j}>0$ holds for all $2\le i\le m$ and $1\le j\le m$. 

Next, it is clear that for $\text{LP}(\boldsymbol{\tau})$ with different $\boldsymbol{\tau}$, all the constraints are the same except for the domain of the variable $\vex$. It is easy to see that for any $\boldsymbol{{\tau}}\in\hat{\mathcal{P}}$, 
\begin{eqnarray*}
\{\vex:\vex\in \prod_{i=1}^m[0,(1-\varepsilon)\hat{\tau}_i]\}&\subseteq& \{\vex:\vex\in \prod_{i=1}^m[0,{\tau}_i]\}\\&\subseteq& \{\vex:\vex\in \prod_{i=1}^m[0,(1+\varepsilon)\hat{\tau}_i]\}
\end{eqnarray*}
Hence, any feasible solution to $\text{LP}((1-\varepsilon)\hat{\boldsymbol{\tau}})$ is also a feasible solution to $\text{LP}(\boldsymbol{\tau})$ where $\boldsymbol{{\tau}}\in\hat{\mathcal{P}}$, and any feasible solution to $\text{LP}(\boldsymbol{\tau})$ is also a feasible solution to  $\text{LP}((1+\varepsilon)\hat{\boldsymbol{\tau}})$. As the ILPs minimize the objective function, we have
$$\text{OPT}_{\text{LP}((1+\varepsilon)\hat{\boldsymbol{\tau}})}\le \text{OPT}_{\text{LP}(\boldsymbol{\tau})}\le \text{OPT}_{\text{LP}((1-\varepsilon)\hat{\boldsymbol{\tau}})}.$$

To prove Lemma~\ref{lemma:domain-1}, it suffices to show that $\text{OPT}_{\text{LP}((1-\varepsilon)\hat{\boldsymbol{\tau}})}$ is not too large compared with $\text{OPT}_{\text{LP}((1+\varepsilon)\hat{\boldsymbol{\tau}})}$.

If $\text{OPT}_{\text{LP}((1+\varepsilon)\hat{\boldsymbol{\tau}})}>(1-c_0\varepsilon)n$, then for sufficiently small $\varepsilon$ it follows directly that 
\begin{eqnarray*}
\text{OPT}_{\text{LP}((1-\varepsilon)\hat{\boldsymbol{\tau}})}&\le& n\\ &\le& \frac{1}{1-c_0\varepsilon}\text{OPT}_{\text{LP}((1+\varepsilon)\hat{\boldsymbol{\tau}})}\\&\le& (1+3c_0\varepsilon)\text{OPT}_{\text{LP}((1+\varepsilon)\hat{\boldsymbol{\tau}})}.
\end{eqnarray*}

Otherwise, $\text{OPT}_{\text{LP}((1+\varepsilon)\hat{\boldsymbol{\tau}})}\le(1-c_0\varepsilon)n$. We make the following claims:
\begin{claim}\label{claim:3}
If $\text{OPT}_{{\text{LP}}((1+\varepsilon)\hat{\boldsymbol{\tau}})}\le(1-c_0\varepsilon)n$, then there exists a near optimal solution $\vex'$ to $\text{OPT}_{{\text{LP}}((1-\varepsilon)\hat{\boldsymbol{\tau}})}$ such that $x'_i\le n\tau_i(1-\varepsilon)$, and $\sum_{i}x_i'\le (1+3c_0\varepsilon) \text{OPT}_{\text{LP}((1+\varepsilon)\hat{\boldsymbol{\tau}})}$. That is, $\vex'$ is a feasible solution to $\text{OPT}_{{\text{LP}}((1-\varepsilon)\hat{\boldsymbol{\tau}})}$.
\end{claim}
\begin{proof}[Proof of Claim~\ref{claim:3}] 
We follow the same proof idea as that of Claim~\ref{claim:1}. For simplicity, let $\vex'$ be the optimal solution to ${\text{LP}}((1+\varepsilon)\hat{\boldsymbol{\tau}})$.  We construct $\vex'$ by modifying $\vex^*$. Let $I=\{1\le i\le m: x_i^*\ge n\hat{\tau}_i(1-\varepsilon)\}$. For $i\in I$, we let $\Delta_i=x_i^*-n\hat{\tau}_i(1-\varepsilon)\ge 0$. Note that since $x_i^*\le n\hat{\tau}_i(1+\varepsilon)$, we have $\Delta_i\le 2\varepsilon\hat{\tau}_i$. For $i\not\in I$, we let  $\Delta_i=n\hat{\tau}_i(1-\varepsilon)-x_i^*\ge 0$. It is easy to see that
\begin{eqnarray*}
\sum_{i\not\in I} \Delta_i-\sum_{i\in I} \Delta_i&=&\sum_i n\hat{\tau}_i(1-\varepsilon)-\sum_ix_i^*\\
&\ge& n(1-\varepsilon)-n(1-c_0\varepsilon)\\
&\ge& (c_0-1)\varepsilon n
\end{eqnarray*}
Meanwhile, $\sum_{i\in I} \Delta_i\le \sum_{i\in I} \varepsilon n\hat{\tau}_i\le \varepsilon n$, hence,
$$\sum_{i\not\in I} \Delta_i-c_0\sum_{i\in I} \Delta_i\ge 0.$$
For $i\not\in I$, pick arbitrary $0\le\Delta_i'\le \Delta_i$ such that 
$$\sum_{i\not\in I} \Delta_i'-c_0\sum_{i\in I} \Delta_i= 0.$$

We define $\vex'$ as follows.
For $i\in I$, we let $x'_i=n\hat{\tau}_i(1-\varepsilon)=x^*_i-\Delta_i$; For $i\not\in I$, we let $x'_i=x_i+\Delta_i$. 
It is easy to see that $x_i'\le n\hat{\tau}_i(1-\varepsilon)$ for all $i$. 

Using exactly the same argument as Claim~\ref{claim:1}, we can prove that $\sum_{j=1}^m (q_{ij}-q_{1j})x_j'\ge n(\lambda_i-\lambda_1)$, and therefore $\vex'$ is feasible to ${\text{LP}}((1-\varepsilon)\hat{\boldsymbol{\tau}})$.

Finally, the same as Claim~\ref{claim:1}, the followings are still true:
$$\text{OPT}_{{\text{LP}}((1+\varepsilon)\hat{\boldsymbol{\tau}})}=\sum_ix_i^*\ge \sum_{i\in I}x_i^*\ge (1-\varepsilon)\sum_{i\in I}n\hat{\tau}_i,$$
and
\begin{eqnarray*}
\sum_{i}x_i'-\sum_ix_i^*&\le& \sum_{i\not\in I} \Delta_i'-\sum_{i\in I}\Delta_i\\
&\le& c_0 \sum_{i\in I}\Delta_i\\
&\le & 2c_0\varepsilon\sum_{i\in I} n\hat{\tau}_i
\end{eqnarray*}

Consequently, it holds that $$\sum_{i}x_i'\le (1+\frac{2c_0\varepsilon}{1-\varepsilon}) \text{OPT}_{{\text{LP}}((1+\varepsilon)\hat{\boldsymbol{\tau}})}\le (1+3c_0\varepsilon) \text{OPT}_{{\text{LP}}((1+\varepsilon)\hat{\boldsymbol{\tau}})}.$$

Since $\vex'$ is a feasible solution to ${{\text{LP}}((1-\varepsilon)\hat{\boldsymbol{\tau}})}$, it follows that $\text{OPT}_{{\text{LP}}((1-\varepsilon)\hat{\boldsymbol{\tau}})}\le \sum_{i}x_i'\le (1+3c_0\varepsilon) \text{OPT}_{{\text{LP}}((1+\varepsilon)\hat{\boldsymbol{\tau}})}$.
\end{proof}
To summarize, it holds that $$\text{OPT}_{{\text{LP}}((1-\varepsilon)\hat{\boldsymbol{\tau}})}\le (1+3c_0\varepsilon) \text{OPT}_{{\text{LP}}((1+\varepsilon)\hat{\boldsymbol{\tau}})}.$$ 

Hence, Lemma~\ref{lemma:domain-1} is true.
\end{proof}

Now we are ready to prove Lemma~\ref{thm:p-main}.
\begin{proof}[Proof of Lemma~\ref{thm:p-main}]

By Lemma~\ref{lemma:tau-1}, we know that $$\sum_{\boldsymbol{\tau'}\in \mathcal{P}\setminus \hat{\mathcal{P}}} \Pr[\boldsymbol{\tau} = \boldsymbol{\tau'}]\leq me^{-\frac{5\varepsilon^2}{6}n}$$

Hence,
\begin{eqnarray*}
&&|\mathbb{E}[\text{OPT}_{\text{LP}(\boldsymbol{{\tau}})}]-\sum_{\boldsymbol{\tau'}\in \hat{\mathcal{P}}} \text{OPT}_{\text{LP}(\boldsymbol{\tau'})}Pr[\boldsymbol{\tau}=\boldsymbol{\tau'}]|\\
&\leq& \left(\max_{\boldsymbol{\tau'}\in \mathcal{P}}\text{OPT}_{{\text{LP}}(\boldsymbol{\tau'})}\right) \sum_{\boldsymbol{\tau'}\in \mathcal{P}\setminus \hat{\mathcal{P}}} \Pr[\boldsymbol{\tau} = \boldsymbol{\tau'}]\\
&\leq& nm e^{-\frac{5\varepsilon^2}{6}n}.
\end{eqnarray*}

Consider an arbitrary $\boldsymbol{\tau'}\in \hat{\mathcal{P}}$. According to Lemma~\ref{lemma:domain-1}, we have

$$|\text{OPT}_{\text{LP}(\boldsymbol{\tau'})}-\text{OPT}_{\text{LP}(\boldsymbol{\hat{\tau}})}| \leq \bar{c}'\varepsilon \text{OPT}_{\text{LP}(\boldsymbol{\hat{\tau}})}$$

Hence, 
\begin{eqnarray*}
&&|\sum_{\boldsymbol{\tau'}\in \hat{\mathcal{P}}} \text{OPT}_{\text{LP}(\boldsymbol{\tau'})}Pr[\boldsymbol{\tau}=\boldsymbol{\tau'}]-\text{OPT}_{\text{LP}(\boldsymbol{\hat{\tau}})}|\\
&\leq& nm e^{-\frac{5\varepsilon^2}{6}n}+ \bar{c}'\varepsilon\text{OPT}_{\text{LP}(\boldsymbol{\hat{\tau}})}
\end{eqnarray*}

To summarize, we have $$|\mathbb{E}[\text{OPT}_{\text{LP}(\boldsymbol{{\tau}})}]-\text{OPT}_{\text{LP}(\boldsymbol{\hat{\tau}})}|\leq 2nm e^{-\frac{5\varepsilon^2}{6}n}+ \bar{c}'\varepsilon\text{OPT}_{\text{LP}(\boldsymbol{\hat{\tau}})}. $$
Notice that $2n e^{-\frac{5\varepsilon^2}{6}n}=\mathcal{O}(\varepsilon)$ when $\varepsilon=\Omega(\sqrt{\frac{\log n}{n}})$, Lemma~\ref{thm:p-main} is true.
\end{proof}

\section{Property of APoLDP}

We know that in order to achieve different levels of privacy, the parameter $\epsilon$ of LDP mechanisms needs to be set to different values. In this section, we study the relationship between APoLDP and the privacy parameter $\epsilon$ for two classes of LDP mechanisms, randomized response and Laplace.

For the voting system with only two candidates, we are able to give the closed-form expression of APoLDP for both randomized response and Laplace. For the voting system with multiple candidates, we can not derive the closed-form expression. But instead, we analyze the maximum value of APoLDP for the randomized response.

\subsection{Two Candidates} 
For $m=2$, it is known that the design matrix $P$ of randomized response and Laplace mechanism can both be expressed as the following systemic form with respect to the parameter $\theta$~\cite{wang2016using}.

$$P(\theta) =\left[
\begin{matrix}
 \frac{1+\theta}{2} & \frac{1-\theta}{2}\\
 \frac{1-\theta}{2} & \frac{1+\theta}{2}
\end{matrix}
\right].$$
When $\epsilon$-LDP is satisfied, 
For randomized response $P^{\epsilon\text{-ran}}=P(\theta^{\epsilon\text{-ran}})$ where $\theta^{\epsilon\text{-ran}}=1-2/(1+e^{\epsilon})$; and for Laplace mechanism $P^{\epsilon\text{-lap}}=P(\theta^{\epsilon\text{-lap}})$ where $\theta^{\epsilon\text{-lap}}=1-e^{-\frac{\epsilon}{2}}$.

According to Eq~\eqref{calc:quv}, we know that matrix $Q=(q_{ij})_{m\times m}$ can also be expressed as $Q(\theta^{\epsilon\text{-ran}})$, $Q(\theta^{\epsilon\text{-lap}})$ where $Q(\theta)$ denotes the following matrix $$ 
Q(\theta)=\left[
\begin{matrix}
 \frac{(1+\theta)\lambda_1}{(1+\theta)\lambda_1+(1-\theta)\lambda_2} &  \frac{(1-\theta)\lambda_1}{(1-\theta)\lambda_1+(1+\theta)\lambda_2}\\
  \frac{(1-\theta)\lambda_2}{(1+\theta)\lambda_1+(1-\theta)\lambda_2} & \frac{(1+\theta)\lambda_2}{(1-\theta)\lambda_1+(1+\theta)\lambda_2}
\end{matrix}
\right].
$$
Further notice that we assume candidate $C_1$ is the designated candidate and $\lambda_2>\lambda_1$, as otherwise there is no need for the attacker to delete voters.

Our Theorem~\ref{thm:main-new} is applicable if for any $1\le j\le 2$,
either $q_{2j}-q_{1j}>0$ or $q_{2j}-q_{1j}<0$, which is essentially always true except when $q_{2j}-q_{1j}=0$ for $j=1$ or $j=2$, that is, when $\theta=\lambda_2-\lambda_1$. Despite that the general proof of Theorem~\ref{thm:main-new} does not hold for $\theta=\lambda_2-\lambda_1$ (as the constants $c_0$ and $c_1$ are not well-defined in this case), it is easy prove that $\text{OPT}_{\overline{\text{LP}}(\boldsymbol{\tau},\delta)}\le \left(1+\mathcal{O}(\delta)\right)\text{OPT}_{{\text{LP}}(\boldsymbol{\tau})}$ and $\text{OPT}_{{\text{LP}}(\boldsymbol{\tau})}\le  \left(1+\mathcal{O}(\delta)\right)\text{OPT}_{\underline{\text{LP}}(\boldsymbol{\tau},\delta)}$ still hold, and furthermore, 
$\mathbb{E}[\text{OPT}_{\text{LP}(\boldsymbol{{\tau}})}]\in [(1-\varepsilon)\text{OPT}_{\text{LP}(\hat{\boldsymbol{{\tau}}})}, (1+\varepsilon)\text{OPT}_{\text{LP}(\hat{\boldsymbol{{\tau}}})}]$ for $\varepsilon=\Omega({\sqrt{\frac{\log n}{n}}})$. Therefore, we will not treat the case $\theta=\lambda_2-\lambda_1$ separately in the following discussion.

To compute APoLDP, it suffices to solve $\text{LP}(\hat{\boldsymbol{{\tau}}})$ directly.

Note that in ${\text{ILP}}(\hat{\boldsymbol{\tau}})$, when $m=2$, Eq~\eqref{mip:1} and Eq~\eqref{mip:2} merge into the following single constraint  \begin{equation} \nonumber
    n\lambda_1-(q_{11}x_1+q_{12}x_2) \geq n\lambda_2 -(q_{21}x_1+q_{22}x_2),
\end{equation}
which is equivalent to \begin{equation} \label{p2candidate-cond}
    (q_{21}-q_{11})x_1+(q_{22}-q_{12})x_2 \geq n(\lambda_2-\lambda_1).
\end{equation}

Recall that $\lambda_2> \lambda_1$. For ease of discussion, let $\phi=\lambda_2-\lambda_1 > 0$. We have the following simple observation.

\begin{lemma}
For any $\theta\in [0,1]$, it holds that $q_{22}-q_{12}\geq 0$ and $ q_{22}-q_{12}\geq q_{21}-q_{11}$.
\end{lemma}

\begin{proof}
Observe that $q_{22}-q_{21}\geq\frac{2\theta \lambda_2}{(1-\theta)\lambda_1+(1+\theta)\lambda_2}\ge 0$. 

Through simple calculation, we know that 
$$\frac{q_{21}-q_{11}}{q_{22}-q_{12}}=\frac{(\lambda_2-\lambda_1)-\theta}{(\lambda_2-\lambda_1)+\theta}\cdot\frac{1+\theta(\lambda_2-\lambda_1)}{1-\theta(\lambda_2-\lambda_1)}$$ $$ < \frac{1-\theta}{1+\theta}\frac{1+\theta}{1-\theta}=1.\qedhere$$
\end{proof}

We know that if $\theta\leq \phi$, then it holds that $q_{21}-q_{11} \geq 0$ and $(q_{22}-q_{12})\hat{\tau}_2\leq \phi$. Hence, the optimal solution of ${\text{LP}}(\hat{\boldsymbol{\tau}})$ is achieved at 
$x_1=n\hat{\tau}_1$, $x_2=n\hat{\tau}_2$.

Otherwise, if $\theta>\phi$, then it holds that $q_{21}-q_{11} < 0$ and $(q_{22}-q_{12})\hat{\tau}_2\leq \phi$. Hence, the optimal solution of ${\text{LP}}(\hat{\boldsymbol{\tau}})$ is achieved at $x_1=0$, $x_2=n\hat{\tau}_2$. 

Combine all, we give the closed form expression of $\text{OPT}_{{\text{LP}}(\hat{\boldsymbol{\tau}})}$ as follow \begin{equation}\nonumber
    \text{OPT}_{{\text{LP}}(\boldsymbol{\hat{\tau}})}=
    \begin{cases}
    n &  \quad\text{if}\ \theta \leq \phi \\
    n \phi\frac{1+\theta\phi}{\theta+\phi}  & \quad\text{otherwise}  \\
    \end{cases}
\end{equation} 

To achieve the $\epsilon$-local differential privacy, the parameter $\theta$ is set to be $\theta^{\epsilon\text{-ran}}=1-2/(1+e^{\epsilon})$ and $\theta^{\epsilon\text{-lap}}=1-e^{-\frac{\epsilon}{2}}$.

For convenience, we denote the voting system as $\mathcal{S}_2=\mathcal{S}_{2,\boldsymbol{\lambda}}(n)$. We denote the randomized response mechanism satisfying $\epsilon$-LDP as $\epsilon$-ran, and Laplace mechanism satisfying $\epsilon$-LDP as $\epsilon$-lap. We know that, for any $\xi\in(0,1)$ it holds that

\noindent\hspace{-20mm}\begin{equation} \label{form:m2p-ran}
    \text{APoLDP}(\mathcal{S}_2,\mathcal{\epsilon\text{-ran}},\xi)\!=\!
    \begin{cases}
    1/\phi &  \text{if}\ \epsilon \leq\ln\frac{1+\phi}{1-\phi} \\
    \frac{e^\epsilon+1+(e^\epsilon-1)\phi}{e^\epsilon-1+(e^\epsilon+1)\phi} & \text{otherwise}  \\
    \end{cases}
\end{equation}

\noindent\hspace{-100mm}\begin{equation}\hspace{-2mm} \label{form:m2p-lap}
    \text{APoLDP}(\mathcal{S}_2,\mathcal{\epsilon\text{-lap}},\xi)\!=\!
    \begin{cases}
    1/\phi &  \!\!\text{if}\ \epsilon \leq 2\ln\frac{1}{1-\phi} \\
    \frac{e^{\epsilon/2}+(e^{\epsilon/2}-1)\phi}{e^{\epsilon/2}-1+e^{\epsilon/2}\phi} & \!\!\text{otherwise}  \\
    \end{cases}
\end{equation} 

\noindent\textbf{Remark.} We can see that $\text{OPT}_{{\text{LP}}(\boldsymbol{\hat{\tau}})}$, and hence PoLDP stays the same once $\theta\le\phi=\lambda_2-\lambda_1$. This coincides with our intuition, because the smaller $\epsilon$ is, the more privacy is enforced. In the extreme case when $\epsilon$ approaches $0$, each voter is almost randomly reporting its type (preference), whereas knowing the reported types does not help the attacker at all. Hence, when the attacker deletes voters, the number of voters of each type decreases proportionally to their fraction $\lambda_j$'s. Given that $\lambda_1<\lambda_2$, unless the attacker delete all voters (and thus $\text{OPT}_{{\text{LP}}(\boldsymbol{\hat{\tau}})}=n$), the probability that the designated candidate becomes one co-winner is negligibly small. 

\subsection{Multiple Candidates}
According to~\cite{wang2016using}, the optimal design matrix for randomized response satisfying $\epsilon$-LDP is given by Eq~\eqref{eq:design-matrix}, from which we can derive the $Q$-matrix as:
\begin{eqnarray}\label{eq:matrix-q}
Q=\begin{bmatrix}
\frac{(\theta+\frac{1-\theta}{m})\lambda_1}{{\vep}_{1} \cdot {\ve\lambda}} & \frac{\frac{1-\theta}{m}\lambda_1}{{\vep}_{2} \cdot {\ve\lambda}} & \cdots & \frac{\frac{1-\theta}{m}\lambda_1}{{\vep}_{m} \cdot {\ve\lambda}} \\
\frac{\frac{1-\theta}{m}\lambda_2}{{\vep}_{1} \cdot {\ve\lambda}} & \frac{(\theta+\frac{1-\theta}{m})\lambda_2}{{\vep}_{2} \cdot {\ve\lambda}} & \cdots & \frac{\frac{1-\theta}{m}\lambda_2}{{\vep}_{m} \cdot {\ve\lambda}} \\
\cdots&\cdots&\cdots&\cdots\\
\frac{\frac{1-\theta}{m}\lambda_m}{{\vep}_{1} \cdot {\ve\lambda}} & \frac{\frac{1-\theta}{m}\lambda_m}{{\vep}_{2} \cdot {\ve\lambda}} & \cdots & \frac{(\theta+\frac{1-\theta}{m})\lambda_m}{{\vep}_{m} \cdot {\ve\lambda}} 
\end{bmatrix}
\end{eqnarray}
where $\vep_i\cdot{\ve\lambda}=\sum_{j}p_{ij}\lambda_j=\frac{1-\theta}{m}+\theta\lambda_i$, and $\theta=1-\frac{m}{m-1+e^{\epsilon}}$.

We have the following simple observation.

\begin{observation}\label{obs:1}
For a voting system $\mathcal{S}$ under randomized response satisfying $\epsilon$-LDP (e.g., the design matrix equals Eq~\eqref{eq:design-matrix}), if $\lambda_1<\lambda_i$ for all $2\le i\le m$, then for any $1\le j\le m$, it holds that $q_{ij}-q_{1j}>0$.
\end{observation}

Given Observation~\ref{obs:1}, we can calculate APoLDP according to Theorem~\ref{thm:main-new}. Through the analysis for $\text{LP}(\boldsymbol{\hat{\tau}})$, we are able to derive the following result.

\begin{theorem}\label{thm:multi-plu}
If the voting system $\mathcal{S}$ under an randomized response mechanism $\mathcal{R}$ satisfies $\epsilon$-LDP and $\lambda_i>\lambda_1$ for all $2\le i\le m$, then for any $\xi\in(0,1)$, APoLDP achieves the maximal value $\frac{1}{1-\sum_{i=2}^m(\lambda_i-\lambda_1)}$ when $\frac{\theta\lambda_{1}}{\lambda_{\max}-\lambda_1}\le \frac{1-\theta}{m}$ where $\lambda_{\max}=\max_i\lambda_i$, i.e., when $\epsilon\le \ln\frac{\lambda_{\max}}{\lambda_1}$. 
\end{theorem}

Theorem~\ref{thm:multi-plu} implies that, by choosing $\epsilon\le \ln\frac{\lambda_{\max}}{\lambda_1}$ we can prevent the voting system from any potential electoral control by deleting voters. The threshold $\ln\frac{\lambda_{\max}}{\lambda_1}$ coincides with the case of $m=2$, in which case by plugging in $\lambda_{\max}=\lambda_2$ and $\phi=\lambda_2-\lambda_1$ we get exactly the threshold of $\ln\frac{1+\phi}{1-\phi}$.

The following part of this subsection is devote to the proof of Theorem~\ref{thm:multi-plu}. We first rewrite $\text{LP}(\boldsymbol{\hat{\tau}})$ by plugging in all the parameters.
\begin{subequations}
\begin{eqnarray}
&\hspace{-10mm}\min& \sum_{i=1}^m x_i \nonumber\\
&&  \frac{\theta\lambda_ix_i}{n\left(\frac{1-\theta}{m}+\theta\lambda_i\right)}- \frac{\theta\lambda_1x_1}{n\left(\frac{1-\theta}{m}+\theta\lambda_1\right)} \nonumber\\  &+& \frac{1-\theta}{m}(\lambda_i-\lambda_1)\sum_{j=1}^m\frac{x_j}{n\left(\frac{1-\theta}{m}+\theta\lambda_j\right)}\le n(\lambda_i-\lambda_1) 
\nonumber\\ &&\hspace{51mm} \forall\; 2\le i\le m \nonumber\\
&&  x_i\leq n\left(\frac{1-\theta}{m}+\theta\lambda_i\right) \qquad\qquad\quad \forall\;
1\le i \le m \nonumber
\end{eqnarray}
\end{subequations}

Define $y_i=\frac{x_i}{n\left(\frac{1-\theta}{m}+\theta\lambda_i\right)}\in [0,1]$, then $\text{LP}(\boldsymbol{\hat{\tau}})$ is equivalent to the following $\text{LP}_{\vey}$:
\begin{subequations}
\begin{eqnarray}
&\hspace{-10mm}\min& \sum_{i=1}^m \left(\frac{1-\theta}{m}+\theta\lambda_i\right)y_i \nonumber\\
&& \left(\frac{\theta\lambda_i}{\lambda_i-\lambda_1}+\frac{1-\theta}{m}\right)y_i-\left(\frac{\theta\lambda_1}{\lambda_i-\lambda_1}-\frac{1-\theta}{m}\right)y_1 \nonumber\\&\ge& 1-\sum_{j\neq 1, i}\frac{1-\theta}{m}y_j  \hspace{20mm}\forall\; 2\le i\le m \label{eq:y1}\\
&&  0\le y_i\le 1 \hspace{29mm} \forall\;
1\le i \le m \nonumber
\end{eqnarray}
\end{subequations}

The proof of Theorem~\ref{thm:multi-plu} relies on the following lemmas.

\begin{lemma}\label{lemma:no-sol}
If $\lambda_i>\lambda_1$ for all $2\le i\le m$, and for some $i=i_0$ it holds that $\frac{\theta\lambda_{1}}{\lambda_{i_0}-\lambda_1}\le \frac{1-\theta}{m}$, then there is only one feasible solution to $\text{LP}_{\vey}$, which is $y_1=y_2=\cdots=y_m=1$.
\end{lemma}
\begin{proof}
When $\frac{\theta\lambda_{1}}{\lambda_{i_0}-\lambda_1}\le \frac{1-\theta}{m}$, we know that 
\begin{eqnarray*}
&&\left(\frac{\theta\lambda_{i_0}}{\lambda_{i_0}-\lambda_1}+\frac{1-\theta}{m}\right)y_{i_0}-\left(\frac{\theta\lambda_1}{\lambda_{i_0}-\lambda_1}-\frac{1-\theta}{m}\right)y_1\\
&\le& \frac{\theta\lambda_{i_0}}{\lambda_{i_0}-\lambda_1}-\frac{\theta\lambda_1}{\lambda_{i_0}-\lambda_1}+\frac{2(1-\theta)}{m}\\
&=&\theta+\frac{2(1-\theta)}{m}
\end{eqnarray*}
In the meantime,
We have that
$$1-\sum_{j\neq 1, i}\frac{1-\theta}{m}y_j \ge 1-\sum_{j\neq 1, i}\frac{1-\theta}{m}=\theta+\frac{2(1-\theta)}{m}.$$
Hence, if Eq~\eqref{eq:y1} is true, then all inequalities become tight and we get $y_1=y_2=\cdots=y_m=1$.
\end{proof}

As a complement to Lemma~\ref{lemma:no-sol}, we have the following.
\begin{lemma}\label{lemma:yes-sol}
If $\lambda_i>\lambda_1$ and $\frac{\theta\lambda_{1}}{\lambda_{i}-\lambda_1}> \frac{1-\theta}{m}$ hold for all $2\le i\le m$, then the optimal objective value of $\text{LP}_{\vey}$ is strictly smaller than $1$. In particular, let 
$\delta=\min_i \{\frac{\theta\lambda_{1}}{\lambda_{i}-\lambda_1}- \frac{1-\theta}{m}\}$, then the optimal objective value of $\text{LP}_{\vey}$ is at most $1-\theta(1-\theta)\delta/2=1-\Omega(\delta)$.
\end{lemma}
\begin{proof}
Consider the following solution:
$y_1'=0$, and
$$y_i'=\left(1-\frac{1-\theta}{m}t\right)\frac{\lambda_i-\lambda_1}{\theta\lambda_i},$$
where $$t=\frac{\sum_{i=2}^m\frac{\lambda_i-\lambda_1}{\theta\lambda_i}}{1+\frac{1-\theta}{m}\sum_{i=2}^m\frac{\lambda_i-\lambda_1}{\theta\lambda_i}}.$$

We first show that $\vey'$ is a feasible solution. To see this, note that by the definition of $t$ it holds that:
$$t= \left(1-\frac{1-\theta}{m}t\right)\sum_{i=2}^m\frac{\lambda_i-\lambda_1}{\theta\lambda_i}.$$
Note that the right side is exactly $\sum_i y_i'$, therefore, 
$$\sum_i y_i'=t.$$
Consequently, 
$$\frac{\theta\lambda_i}{\lambda_i-\lambda_1}y_i'=1-\frac{1-\theta}{m}t=1-\frac{1-\theta}{m}\sum_i y_i',$$
implying that Eq~\eqref{eq:y1} is satisfied. To summarize, $\vey'$ is feasible to $\text{ILP}_{\vey}$. 

Next, we estimate the value of $t$. Since $\frac{\theta\lambda_{1}}{\lambda_{i}-\lambda_1}- \frac{1-\theta}{m}\ge \delta$, we have that
$$\frac{\lambda_i-\lambda_1}{\theta\lambda_i}\le \frac{m}{1-\theta+(1+\delta)m\theta}, \quad \forall i\ge 2$$
Consequently,
$$\sum_{i=2}^m\frac{\lambda_i-\lambda_1}{\theta\lambda_i}\le \frac{m(m-1)}{1-\theta+(1+\delta)m\theta}.$$
Plugging it into the definition of $t$, we get
$$t\le \frac{\frac{m(m-1)}{1-\theta+(1+\delta)m\theta}}{1+\frac{1-\theta}{m}\frac{m(m-1)}{1-\theta+(1+\delta)m\theta}}=\frac{m-1}{\delta\theta+1}.$$

Finally, we estimate the objective value of $\vey'$. We have that
\begin{eqnarray*}
\sum_{i=1}^m \left(\frac{1-\theta}{m}+\theta\lambda_i\right)y_i'&=&\frac{1-\theta}{m}t+\theta\sum_{i=1}^m\lambda_iy_i'\\&\le& \frac{1-\theta}{m}t+\theta\\&=&\frac{1+\theta^2\delta-\frac{1-\theta}{m}}{1+\theta\delta}\\
&\le& \frac{1+\theta^2\delta}{1+\theta\delta}
\le 1-\frac{\theta(1-\theta)}{2}\delta
\end{eqnarray*}


\end{proof}

Plug in the definition of $\theta=1-\frac{m}{m-1+e^{\epsilon}}$ for $\epsilon$-LDP, we can directly have Theorem~\ref{thm:multi-plu}.

\section{Numeric Experiments}

In this section, we demonstrate PoLDP through numeric experiments\footnote{All experiments are performed on a computer with an Intel i7-6700 processor and 32GB memory. The code is available at https://github.com/polyapp/poldp}. In our experiments, the number of voters is set to be $n=10^8$. The number of candidates is set to be $m=2;5$. In this experiment, we generate the true type of each voter using two methods. 

\begin{itemize}
    \item The first method guarantees the difference in score between the designated candidate and the winner equals $m\phi$ where the true type of each voter is randomly generated from $[1,m]$. 
    
    \item The second method randomly generates the true type of each voter according to the real-world data---Sushi Data~\cite{kamishima2003nantonac} which is a commonly used data set for generating preferences~\cite{azari2012random}.
\end{itemize}

For ease of notation, we denote by $\mathcal{T}_{m}^{\phi}$ the voting system generated by the first method, and $\mathcal{T}_{m}$ the voting system generated by the second method. For each kind of voting systems, we generate $2000$ instances. 

The parameter is set to be $\xi=0.999$, $\delta=0.001$, $\epsilon=0.001$. 
We observe that randomized response and Laplace mechanism on our randomly generated voting systems hold that:

$$\text{avg}[\frac{|\text{OPT}_{\overline{\text{LP}}((1-\epsilon)\boldsymbol{\hat{\tau}},\delta)}-\text{OPT}_{\text{LP}(\boldsymbol{\hat{\tau}})}|}{\text{OPT}_{\text{LP}(\boldsymbol{\hat{\tau}})}}]\leq 5\%,$$
and 
$$\text{avg}[\frac{|\text{OPT}_{\underline{\text{LP}}((1+\epsilon)\boldsymbol{\hat{\tau}},\delta)}-\text{OPT}_{\text{LP}(\boldsymbol{\hat{\tau}})}|}{\text{OPT}_{\text{LP}(\boldsymbol{\hat{\tau}})}}]\leq 5\%.$$

This result suggests if we compare two solutions, one guarantees that the designated candidate wins by at least 99.9\%, and the other guarantees that the designated candidate wins by at most 0.1\%, then the two manipulation costs differ by at most $10\%$. In other words, on the voting systems we created $\text{PoLDP}(\mathcal{S},\mathcal{R},0.1\%)$ and $\text{PoLDP}(\mathcal{S},\mathcal{R},99.9\%)$ differ by at most $10\%$ in average. Based on this, we simply calculate the manipulation cost with LDP mechanisms using $\text{OPT}_{\text{LP}(\boldsymbol{\hat{\tau}})}$.

For voting systems with two candidates, we perform the experiments on voting systems $\mathcal{T}^{\phi}_2$ for $\phi=0.1;0.2;0.3;0.4$. We know that once $\phi$ is fixed, all instances in $\mathcal{T}^{\phi}_2$ have the same $\boldsymbol{\lambda}$. For $\epsilon=0,0.01,\dots,0.49,5$, we calculate the corresponding PoLDP of randomized response and Laplace mechanism and plot it in Figure~\ref{fig:poldp-m2p}. We observe that the curve ``PoLDP - $\epsilon$'' given by numeric experiments matches the mathematical closed-form expression of APoLDP given by Eq~\eqref{form:m2p-ran}, Eq~\eqref{form:m2p-lap}.

\begin{figure}[!ht]
	\begin{center}
		\includegraphics[width=85mm,height=62.5mm]{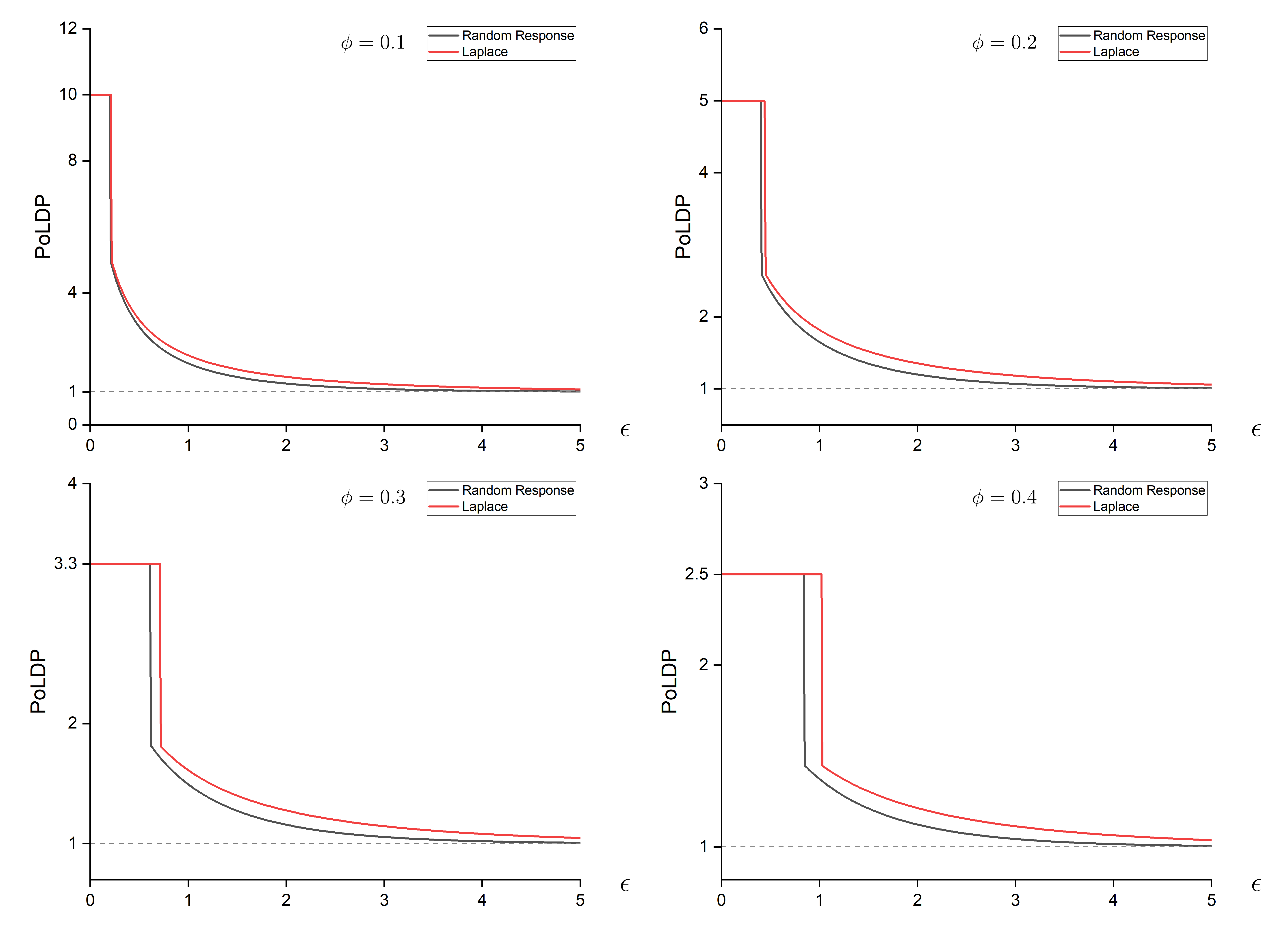}
		\caption{PoLDP of randomized response and Laplace mechanism on voting systems $\mathcal{T}_2^{\phi}$.}
		\label{fig:poldp-m2p}	
		\end{center}
\end{figure}

For voting systems with multiple candidates, we perform experiments on voting systems $\mathcal{T}_5$ and $\mathcal{T}_5^{\phi}$ for $\phi=0.1;0.2;0.3;0.4$. For $\epsilon=0,0.01,\dots,0.49,5$, we compare the efficiency of randomized response and Laplace mechanism from two aspects: the average and the standard deviation of PoLDP (solid lines represent the average, and dotted shadows represent the standard deviation); the percentage of instances which has PoLDP larger than 99\% of its maximal value (in short, the maximal percentage). We summarize the comparison in Figure~\ref{fig:poldp-m5p}-\ref{fig:poldp-secure-m5p-sushi}. 

\begin{figure}[!ht]
	\begin{center}
		\includegraphics[width=85mm,height=62.5mm]{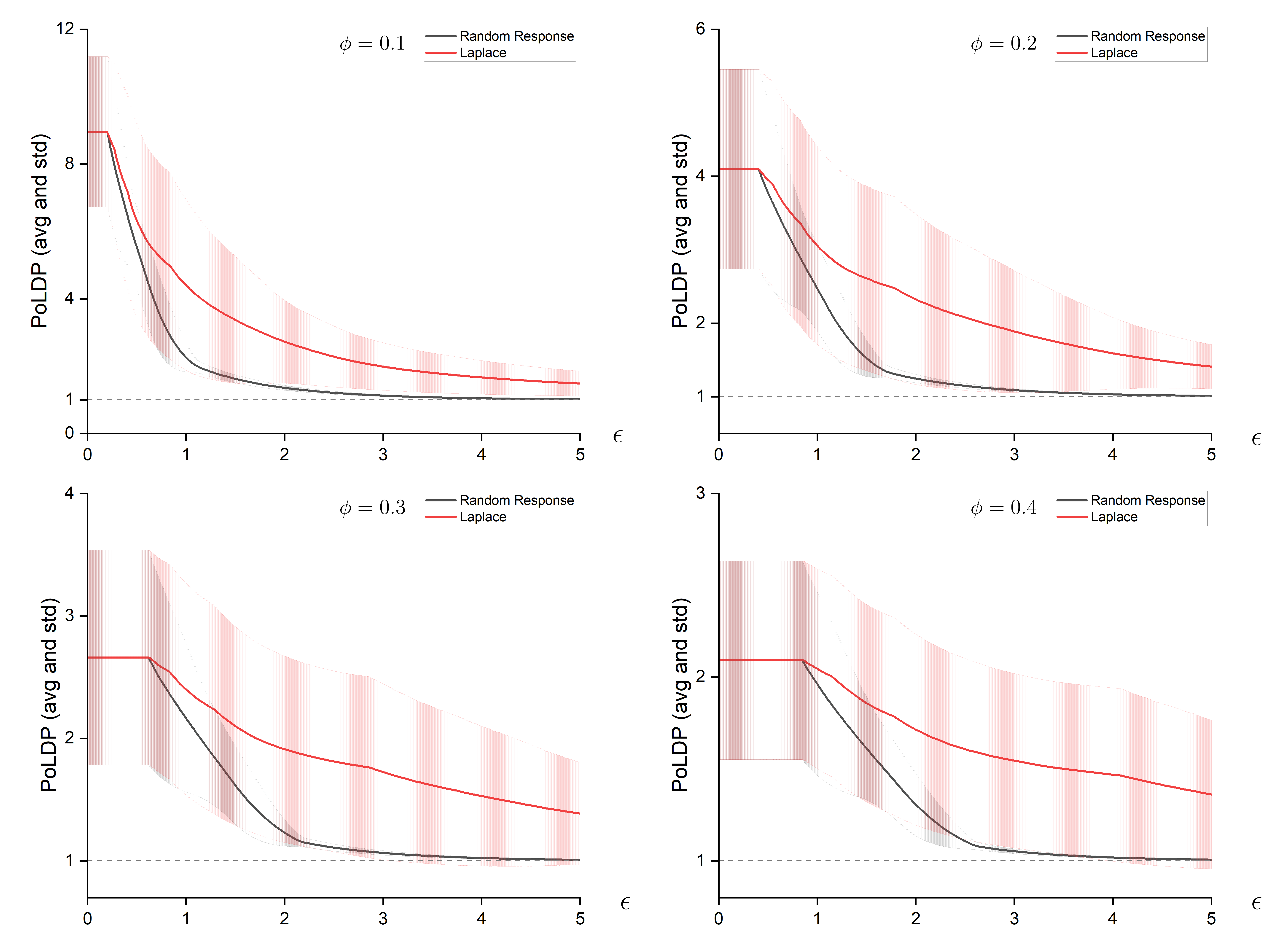}
		\caption{PoLDP of randomized response and Laplace mechanism on voting systems $\mathcal{T}_5^{\phi}$.}
		\label{fig:poldp-m5p}	
		\end{center}
\end{figure}

\begin{figure}[!ht]
	\begin{center}
		\includegraphics[width=85mm,height=59.31mm]{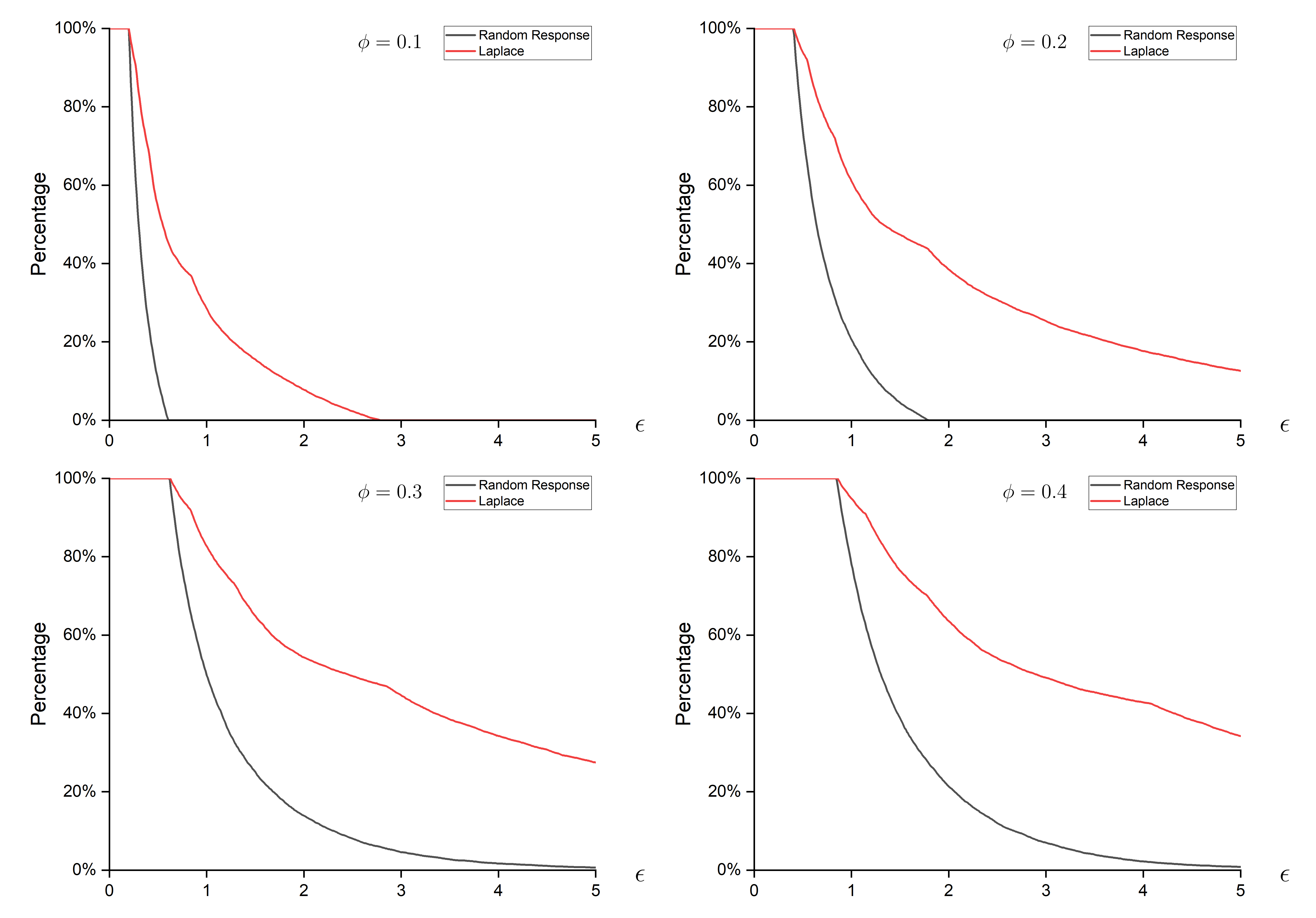}
		\caption{The maximal percentage of randomized response and Laplace mechanism on voting systems $\mathcal{T}_5^{\phi}$.}
		\label{fig:secure-m5p}	
		\end{center}
\end{figure}

\begin{figure}[!ht]
	\begin{center}
		\includegraphics[width=85mm,height=31mm]{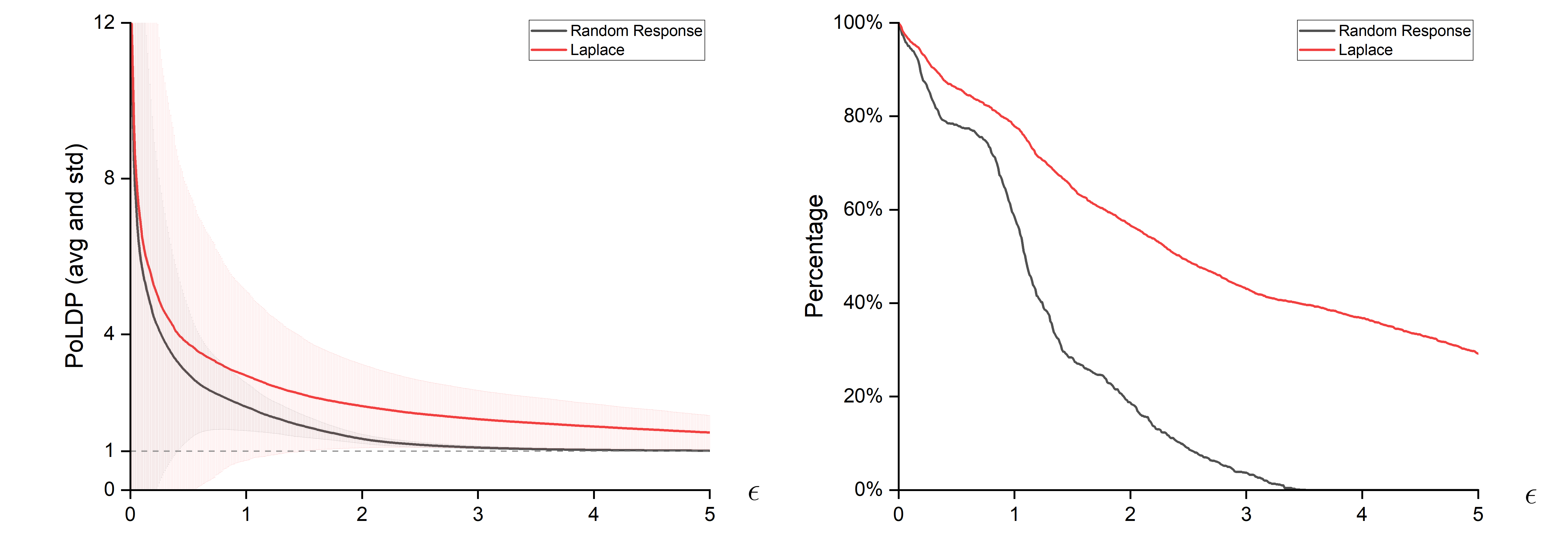}
		\caption{PoLDP and the maximal percentage of randomized response and Laplace mechanism on voting systems $\mathcal{T}_5$.}
		\label{fig:poldp-secure-m5p-sushi}	
		\end{center}
\end{figure}

We observe that the Laplace mechanism outperforms randomized response by always giving a larger PoLDP on all the voting systems we test. Moreover, the randomized mechanism is more sensitive to the selection of $\epsilon$ on voting systems $\mathcal{T}_5^{\phi}$. We can see from Figure~\ref{fig:poldp-m5p}, the average of PoLDP for randomized response decreases quicker than Laplace in the region $\epsilon\in[0,2]$. Interestingly, we also observe that the maximal percentage coincides with the value of PoLDP (e.g., the maximal percentage equals 100\% where PoLDP stays at its maximal value) on voting systems $\mathcal{T}_5^{\phi}$.

\section{Discussion}

In real-world applications, although the exact value of ${\ve\lambda}$ is unknown, we can use a two-round implementation of LDP in voting to estimate ${\ve\lambda}$: In the first round, a random subset $\mathcal{A}$ of voters are selected to report their types via the LDP mechanism; In the second round, all voters except those in $\mathcal{A}$, report their types via the LDP mechanism. We use the reported types given by voters in $\mathcal{A}$, to get an estimation of ${\ve\lambda}$. It is generally sufficient to obtain a good statistical estimation through a reasonably small sample~\cite{kenny1986statistics}. In particular, $|\mathcal{A}|$ can be $o(n)$ for sufficiently large $n$, and therefore it will only marginally affect the value of PoLDP.    

\section{Conclusion}

We provide the first systematic study on how LDP mechanisms may improve the resilience of a voting system under electoral control by deleting voters. The metric PoLDP introduced in this paper gives a general guidance towards the choice of the privacy parameter $\epsilon$ in LDP mechanisms.

\section*{Acknowledgments}

Liangde Tao is partly supported by ``New Generation of AI 2030'' Major Project (2018AAA0100902). Lin Chen is partly supported by NSF Grant 2004096.

\bibliographystyle{named}
\bibliography{ijcai22}

\end{document}